\RequirePackage{fix-cm}
\documentclass[final]{svjour3}                     

\usepackage{fullpage}
\usepackage{graphicx}
\usepackage{xcolor}

\usepackage{amsmath}
\usepackage{amssymb}
\usepackage{latexsym}

\def\setN{\mathbb{N}}
\def\obs{\textit{obs}}

\def\ie{i.e.,~}

\newcommand{\PSPACE}{\textbf{\rm PSPACE}}
\newcommand{\EXPSPACE}{\textbf{\rm EXPSPACE}}

\newcommand{\word}[1]{\langle #1\rangle}
\newcommand{\set}[1]{\{#1\}}

\newcommand\LT[1]{\xrightarrow{#1}}

\newcommand{\view}{{\tt view}}
\newcommand\nxt{\mathit{next}}

\newcommand\power[1]{\mathcal{P}(#1)}
\newcommand\Run[1]{\mathcal{R}(#1)}
\newcommand\A{\mathcal{A}}

\newcommand\tview{\view}
\newcommand\M{\mathbb{M}}

\newcommand{\proj}{proj}  
\newcommand{\syncmachines}{\M^s}
\newcommand{\run}{r}
\newcommand{\vw}{v}
\newcommand{\tNDI}{{\tt NDI}}
\newcommand{\tNDS}{{\tt NDS}}
\newcommand{\tRES}{{\tt RES}}
\newcommand{\NDIs}{\tNDI}
\newcommand{\NDSs}{\tNDS}
\newcommand{\pvh}{{\cal V}_H}
\newcommand{\sched}{\mathit{sched}}
\newcommand{\Aut}{\mathcal{\A}}
\newcommand{\intrel}{\rightarrow}
\newcommand{\commentout}[1]{}
\newcommand{\NSPACE}{\textbf{{\rm NSPACE}}}
\newcommand{\DSPACE}{\textbf{{\rm DSPACE}}}

\usepackage{listings}

\renewcommand{\emptyset}{\varnothing}

\lstdefinelanguage{Algo}{%
commentstyle={\color{gris}\it},
morekeywords={Initialization,Main,while,do,if,then,else,endif,endwhile},
keywordstyle={\sffamily  \textbf},%
sensitive=false,%
texcl=true,
flexiblecolumns=true,
morecomment=[l][\it\color{gris}]//,%
morecomment=[s][\it\color{gris}]{/*}{*/},
showstringspaces=true,
}

\usepackage[ruled]{algorithm2e}
\renewcommand*\thelstnumber{${\the\value{lstnumber}}\!\!:$}
\newcommand{\lline}[1]{\the\value{#1}}

\SetKwInput{PreCond}{Pre-condition}
\SetKwInput{PostCond}{Post-condition}

\journalname{ }

\begin{document}

\title{The Complexity of Synchronous Notions \\of Information Flow Security}

\author{Franck Cassez
       \and Ron van der Meyden
       \and Chenyi Zhang\thanks{Supported by Australian Research Council Discovery grant DP1097203.}
}

\institute{Franck Cassez \at
              National ICT Australia, Sydney, Australia \\
              \email{{\tt Franck.Cassez@nicta.com.au}}           
           \and
           Ron van der Meyden \at
              School of Computer Science and Engineering,\\
              University of New South Wales, Sydney, Australia\\
              \email{{\tt meyden@cse.unsw.edu.au}}
           \and
           Chenyi Zhang \at
             School of Computer Science and Engineering,\\
             University of New South Wales, Sydney, Australia \\
             \email{{\tt chenyi@uq.edu.au}}\\
             \emph{Current affiliation:} 
             School of Information Technology and Electrical Engineering,\\ 
             University of Queensland, Brisbane, Australia
}

\date{Received: date / Accepted: date}

\maketitle

\begin{abstract}
The paper considers the complexity of verifying that a finite state system
satisfies a number of definitions of information flow security.
The systems model considered is one in which agents operate
synchronously with awareness of the global clock. This enables timing
based attacks to be captured, whereas previous work on this topic
has dealt primarily with asynchronous  systems.
Versions of the notions of
nondeducibility on inputs, nondeducibility on strategies, and
an unwinding based notion are formulated for this model.
All three notions are shown to be decidable, and their computational
complexity is characterised.
\end{abstract}

\section{Introduction}

Information flow security is concerned with the ability
of agents in a system to deduce information about the activity and
secrets of other agents. An information flow security policy prohibits
some agents from knowing information about other agents. In an
insecure system, an agent may nevertheless be able to make inferences
from its observations, that enables it to deduce facts that it is not
permitted to know. In particular, a class of system design flaws,
referred to as \emph{covert channels}, provide unintended ways for
information to flow between agents, rendering a system insecure.

Defining what it is for a system to satisfy an information flow
security policy has proved to be a subtle matter.  A substantial
literature has developed that provides a range of formal systems
models and a range of definitions of security.  In particular, in
non-deterministic systems it has been found necessary to clarify the
attack model, and distinguish between a passive attacker, which merely
aims to deduce secret information from observations it is able to make
from its position outside the security domain to be protected, and a
more active attacker, that may have planted a Trojan Horse in the
domain to be protected, and which seeks to use covert channels to pass
information out of this domain. While this distinction turns out not
to matter in asynchronous systems~\cite{FG95}, in synchronous settings,
it leads to two different definitions of security, known as
Nondeducibility on Inputs ($\tNDI$)  \cite{sutherland_86}, and Nondeducibility on Strategies
($\tNDS$) \cite{WJ90}.  (The term strategies in the latter refers to the
strategies that a Trojan Horse may employ to pass information out of
the security domain.)  Considerations of proof methods for security,
and compositionality of these methods, has led to the introduction of
further definitions of security, such as {\em unwinding
  relations} \cite{GM84} and the associated definition of {\em restrictiveness}
($\tRES$) \cite{McCullough88}.

One of the dimensions along which it makes sense to evaluate a
definition of security is the practicality of verification techniques
it enables. The early literature on the topic was motivated primarily by 
theorem proving verification methods, but in recent years the feasibility 
of automated verification techniques has begun to be investigated~\cite{DHKRS08,FG95,FG96,FGM00,KB06,MZ06b}. 
This recent work on automated verification of security has dealt primarily 
with asynchronous systems models.

\paragraph{\bfseries Our Contribution}
In this paper we investigate the complexity of
automated verification for a range of definitions of information flow
in a {\em synchronous} systems model, in which agents are aware of a
global clock and may use timing information in their deductions. This
model is significant in that a number of timing-based attacks have
been demonstrated that, e.g., enable cryptographic keys to be deduced
just from the amount of time taken to perform cryptographic
operations~\cite{Kocher96}.  It is therefore desirable that systems
designs are free of timing-based covert channels; the asynchronous
definitions of security that have been the focus of much of the
literature fail to ensure this.

We study three definitions of security in this paper: synchronous
versions of Nondeducibility on Inputs ($\tNDI$), Nondeducibility on
Strategies ($\tNDS$) and an unwinding based definition called
Restrictiveness ($\tRES$). We consider just a two-agent setting, with
agents $L$ for a low security domain and $H$ for a high security
domain, and the (classical) security policy that permits $H$ to know
about $L$'s activity, but prohibits $L$ from knowing about the
activity of $H$. We show that all three definitions are decidable in
finite state systems, and with complexities of PSPACE-complete for
$\tNDI$, EXPSPACE-complete for $\tNDS$, and polynomial time for
$\tRES$.  A preliminary version of this paper,
with only proof sketches, appeared
in~\cite{CvdMZ2010}.  In this extended version,
we provide detailed proofs for all results.

\paragraph{\bfseries Outline of the paper}
The structure of the paper is as follows. Section~\ref{sec:semdef}
introduces our systems model, the definitions of security that we
study, and states the main results of the paper.  The following
sections discuss the proofs of these results.  Section~\ref{sec:ndi}
deals with Nondeducibility on Inputs, Section~\ref{sec:nds} deals with
Nondeducibility on Strategies, and Section~\ref{sec:res} deals with
the unwinding-based definition.  Related literature is discussed in
Section~\ref{sec:related}, and Section~\ref{sec:concl} makes some
concluding remarks.

\section{Notations, Semantic Model and Information Flows Security
  Policies}\label{sec:semdef}

\subsection{Notation}
Otherwise stated, we use standard notation from automata theory.
Given a finite set (alphabet) $A$, we write $A^*$ for the set of finite words over $A$. 
We denote the empty word by $\epsilon$, and for $w \in A^*$, we write $|w|$ for the length of $w$.  
For $n \in \setN$, $A^n$ stands for the set of words of length $n$ over $A$.

\subsection{Synchronous Machines}

We work with a synchronous, non-deterministic state machine model for
two agents, $H$ and $L$. At each step of the computation, the agents
(simultaneously) perform an action, which is resolved
non-deterministically into a state transition. Both agents make
(possibly incomplete) observations of the state of the system, and do
so with awareness of the time. Time is discrete and measured by the
number of steps in a computation.

Our machine model is given in the following definition.  We do not
make any finiteness assumptions in this section and the results in
this section hold for this unconstrained model.
\begin{definition}[Synchronous Machine] \label{def-sync-mac}
  A {\em synchronous machine} $M$ is a tuple of the form $\word{S, A,
    s_0, \rightarrow,$ $ O, \obs}$ where
  \begin{itemize}
  \item $S$ is the set of states,
  \item $A=A_H\times A_L$ is a set of joint actions (or joint inputs),
    each composed of an action of $H$ from the set $A_H$ and an action
    of $L$ from the set $A_L$,
  \item $s_0$ is the initial state,
  \item $\rightarrow \subseteq S \times A\times S$ defines state
    transitions resulting from the joint actions,
  \item $O$ is a set of observations,
  \item $\obs : S\times\set{H, L} \rightarrow O$ represents the
    observations made by each agent in each state.
  \end{itemize}
  We write $\obs_u$ for the mapping $\obs(\cdot,u): S \rightarrow O$,
  and $s\xrightarrow{\ a \ } s'$ for $\langle s,a,s'\rangle\in
  \rightarrow$.  We assume that machines are {\em input-enabled}, by
  requiring that for all $s\in S$ and $a\in A$, there exists $s'\in S$
  such that $s\xrightarrow{\ a \ } s'$. We write $\syncmachines$ for
  the set of synchronous machines.
\end{definition}
A \emph{run} $\run$ of $M$ is a finite sequence $\run = s_0 a_1
s_1\ldots a_n s_n$ with: $a_i \in A$ and $s_i\xrightarrow{a_{i+1}}
s_{i+1} $ for all $i= 0\ldots n-1$.  We write $\Run{M}$ for the set of
all runs of $M$.  We denote the sequence of joint actions $a_1\ldots
a_n$ in the run $\run$ by $Act(\run)$.  For each agent $u \in \{H,L\}$
we define $\proj_u : A \rightarrow A_u$ to be the projection of joint
actions onto agent $u$'s actions. We write $Act_u(\run)$ for the
sequence of agent $u$'s actions in $Act(\run)$, e.g., if $Act(\run) =
a_1\ldots a_n$ then $Act_u(\run) = \proj_u(a_1)\ldots \proj_u(a_n)$.

\subsection{Agent Views}
For a sequence $w$, and $1\leq i \leq |w|$, we write $w_i$ for the
$i$-th element of $w$, and $w[i]$ for the prefix of $w$ up to the
$i$-th element.  We assume agents have a {\em synchronous} view of the
machine, making an observation at each moment of time and being aware
of each of their own actions (but not the actions of the other agent,
which are given simultaneously and independently).  Given a
synchronous machine $M$, and $u \in \{H,L\}$, we define \emph{$u$
views} by the mapping $\tview_u: \Run{M} \rightarrow O(A_{u} O)^*$ by:
$$\tview_u(s_0a_1s_1a_2 \cdots a_ns_n) = \obs_u(s_0)\, \proj_u(a_1)\, \obs_u(s_1)\,
\proj_u(a_2) \cdots \proj_u(a_n)\, \obs_u(s_n) \mathpunct.$$
Intuitively, this says that an agent's view of a run is the history of
all its state observations as well as its own actions in the run.  We
say that a sequence $\vw$ of observations and actions is a {\em
  possible $u$ view in a system $M$} if there exists a run $r$ of $M$
such that $\vw = \view_u(r)$.  The mapping
$\tview_u$ extends straightforwardly to sets of runs $R\subseteq
\Run{M}$, by $\tview_u(R) = \{\tview_u(r)~|~r\in R\}$.  We define the
length $|\vw|$ of a view $\vw$ to be the number of actions it
contains.

\subsection{Expressiveness Issues} \label{sec:schedmachine}

We remark that the model is sufficiently expressive to represent an
alternate model in which agents act in turn under the control of a
scheduler. We say that a synchronous machine is {\em scheduled} if for
each state $s\in S$ either
\begin{itemize}
\item for all actions $a\in A_H$ and $b,b'\in A_L$, and states $t\in S$,
      $s\xrightarrow{(a,b)} t$ iff $s\xrightarrow{(a,b')} t$, or
\item for all actions $a,a'\in A_H$ and $b\in A_L$, and states $t\in S$,
      $s\xrightarrow{(a,b)} t$ iff $s\xrightarrow{(a',b)} t$.
\end{itemize}
This definition says that state transitions in a scheduled machine are
determined by the actions of \emph{at most one} of the agents (the
agent scheduled at that state); the other agent has no control over
the transition.  The model involving machines under the control of a
scheduler of~\cite{MZ08}, in which at most one agent acts at each step
of the computation, can be encoded as scheduled synchronous machines.

\subsection{Notions of Information Flow Security}
We consider a number of different notions of information flow
security.  Each definition provides an interpretation for the security
policy $L \intrel H$, which states that information is permitted to
flow from $L$ to $H$, but not from $H$ to $L$.  Our definitions are
intended for synchronous systems, in which the agents share a clock
and are able to make deductions based on the time.  (Much of the prior
literature has concentrated on asynchronous systems, in which an agent
may not know how many actions another agent has performed.)

\subsubsection{Non-Deducibility on Inputs}

The first definition we consider states that $L$ should not be able to infer $H$
actions from its view.

\begin{definition} A synchronous machine $M$ satisfies
  \emph{Non-Deducibility on Inputs} ($M\in\NDIs$) if for every possible $L$
  view $\vw$ in $M$ and every sequence of $H$ actions $\alpha\in
  A_H^*$ with $|\alpha|=|\vw|$,
  there exists a run $\run\in\Run{M}$ such that $Act_H(\run)=\alpha$
  and $\view_L(\run)=\vw$.
\end{definition}

Intuitively, in a synchronous system, $L$ always knows how many
actions $H$ has performed, since this is always identical to the
number of actions that $L$ has itself performed.  In particular, if
$L$ has made view $\vw$, then $L$ knows that $H$ has performed $|\vw|$
actions. The definition says that the system is secure if this is {\em
  all} that $L$ can learn about what sequence of actions $H$ has
performed.  Whatever $L$ observes is consistent with any sequence of
actions by $H$ of this length\footnote{Recall that $M$ is
  input-enabled.}.  More precisely, define $K_L(\vw)$ for an $L$ view
$\vw$ to be the set of $H$ action sequences $Act_H(r)$ for $r$ a run
with $\vw=\view_L(r)$; this represents what $L$ knows about $H$'s
actions in the run.  Then $M\in\tNDI$ iff for all possible $L$ views
$\vw$ we have $K_L(\vw) = A_H^{|\vw|}$.

The definition of $\NDIs$ takes the viewpoint that a system is secure
if it is not possible for $L$ to make any nontrivial deductions about
$H$ behaviour, provided that $H$ does not actively seek to communicate
information to $L$. This is an appropriate definition when $H$ is
trusted not to deliberately act so as to communicate information to
$L$, and the context is one where $H$ is equally likely to engage in
any of its possible behaviours.  In some circumstances, however,
$\NDIs$ proves to be too weak a notion of security. In particular,
this is the case if the attack model against which the system must be
secure includes the possibility of Trojan Horses at the $H$ end of the
system, which must be prevented from communicating $H$ secrets to $L$.
The following example, due in essence to Wittbold and Johnson
\cite{WJ90} shows that it is possible for a system to satisfy $\NDIs$,
but still allow for $L$ to deduce $H$ information.

\begin{figure}
\centering
    \includegraphics[scale=0.94]{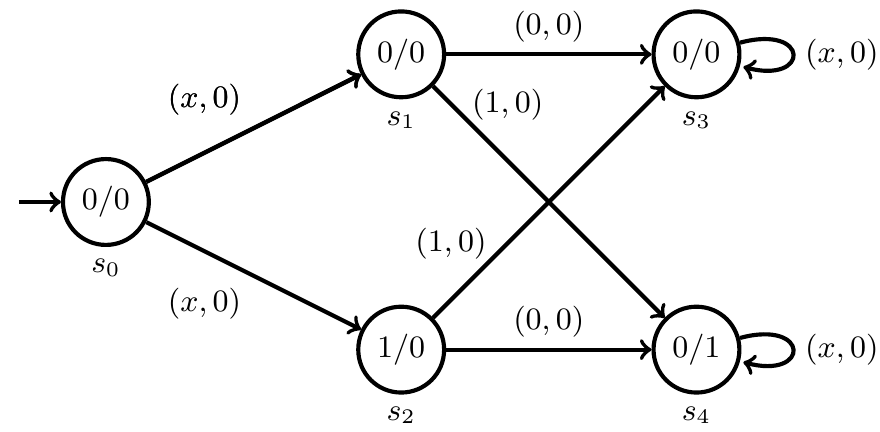}
\caption{A synchronous machine in $\NDIs$, but not in $\NDSs$, where $x\in\set{0,1}$.}
\label{fig:nds}
\end{figure}

\begin{example}\label{example:nds-ndi}
  We present a synchronous machine that satisfies $\NDIs$ in
  Fig.~\ref{fig:nds}.  We use the convention in such figures that the
  observations are shown on a state $s$ in the form of
  $\obs_H(s)/\obs_L(s)$.  Edges are labelled with joint actions
  $(a,a')$ where $a\in A_H$ and $a'\in A_L$. When $a$ is $x$ this means
  that there is such an edge for all $a\in A_H$.  In this example the
  action sets are $A_H=\set{0,1}$, $A_L=\set{0}$.  Note that in state
  $s_1$ and $s_2$, $L$'s observation in the next state is determined
  as the \emph{exclusive-or} of $H$'s current observation and $H$'s
  action. The system is in $\NDIs$ since every $H$ action sequence is compatible
  with every $L$ view of the same length.  For example, the
  $L$ view $00000$ is consistent with $H$ action
  sequence $00$ and $10$ 
  (path $s_0s_1s_3$) and with $H$ action sequence $01$ and $11$ 
  (path $s_0s_2s_3$). Nevertheless, $H$ can communicate a bit $b$ of
  information to $L$, as follows.  Note that $H$ is able to
  distinguish between state $s_1$ and $s_2$ by means of the
  observation it makes on these states (at time 1).  Suppose $b=1$,
  then $H$ chooses action $1$ at $s_1$ and action $0$ at $s_2$; in
  either case the next state is $s_4$, and $L$ observes $1$.
  Alternately, if $b=0$, then $H$ chooses action $0$ at $s_1$ and
  action $1$ at $s_2$; in either case the next state is $s_3$, and $L$
  observes $0$. Whatever the value of $b$, $H$ has guaranteed that $L$
  observes $b$ at time 2, so this bit has been
  communicated. Intuitively, this means that the system fails to
  block Trojan Horses at $H$ from communicating with $L$, even though
  it satisfies $\tNDI$.  (The structure can be repeated so that $H$
  can communicate a message of any length to $L$ in plain text.)  \qed
\end{example}

\subsubsection{Non-Deducibility on Strategies}
The essence of Example~\ref{example:nds-ndi} is that $L$ is able to
deduce $H$ secrets based not just on its knowledge of the system, but
also its knowledge that $H$ is following a particular strategy for
communication of information to $L$. In response to this example,
Wittbold and Johnson proposed the following stronger definition of
security that they called {\em non-deducibility on strategies}.
To state this definition, we first formalize the possible communication strategies that
can be used by $H$. Intuitively, $H$'s behaviour may %
depend on what $H$ has been able to observe in the system.

\begin{definition}[$H$ Strategy, Consistent Runs]
  An {\em $H$ strategy} in $M$ is a function $\pi : \tview_H(\Run{M})
  \rightarrow A_H$ mapping each possible view of $H$ (in $M$) to an
  $H$ action. A run $r=s_0 a_1 s_1 \ldots a_n s_n$ of $M$ is {\em
    consistent with} an $H$ strategy $\pi$ if for all $i=0\dots n-1$,
  we have $\proj_H(a_{i+1}) = \pi(\view_H(s_0 a_1 s_1 \ldots a_{i}
  s_{i}))$.  We write $\Run{M,\pi}$ for the set of runs of $M$ that
  are consistent with the $H$ strategy $\pi$.
\end{definition}
We can now state Wittbold and Johnson's definition.

\begin{definition}
  A synchronous system $M$ satisfies {\em Nondeducibility on
    Strategies} ($M\in \NDSs$), if for all $H$ strategies $\pi_1,
  \pi_2$ in $M$, we have $\view_L(\Run{M,\pi_1}) =
  \view_L(\Run{M,\pi_2})$.
\end{definition}

Intuitively, this definition says that the system is secure if $L$ is
not able to distinguish between different $H$ strategies by means of
its views.  In Example~\ref{example:nds-ndi}, 
given an $H$ strategy $\pi_1$ satisfying $\pi_1(0x0)=0$ and
$\pi_1(0x1)=1$, and another $H$ strategy $\pi_2$ satisfying
$\pi_2(0x0)=1$ and $\pi_2(0x1)=0$, we have the $L$ view $00001$ in
$\view_L(\Run{M,\pi_2}$ but not in $\view_L(\Run{M,\pi_1})$. %
Thus, the sets of $L$ views differ for these two strategies, so the
system is not in $\tNDS$.

An alternate formulation of the definition can be obtained by noting
that for every possible $L$ view $\vw$, there is an $H$ strategy $\pi$
such that $\vw \in \view_L(\Run{M,\pi})$, viz., if $\vw = \view_L(r)$,
we take $\pi$ to be a strategy that always performs the same action at
each time $i<|r|$ as $H$ performs at time $i$ in $r$.  Thus, we can
state the definition as follows:

\begin{proposition}\label{prop:charnds}
$M\in \NDSs$ iff  for all $H$ strategies $\pi$ in $M$, we have
 $\view_L(\Run{M,\pi}) = \view_L(\Run{M})$.
\end{proposition}

\begin{proof}
  It is trivial that if $\view_L(\Run{M,\pi}) = \view_L(\Run{M})$ for
  all strategies $\pi$ then $M\in \NDSs$.  Conversely, suppose that
  $M\in \NDSs$, and let $\pi$ be any strategy.  Plainly
  $\view_L(\Run{M,\pi}) \subseteq \view_L(\Run{M})$; we show the
  reverse containment.  Let $\vw\in \view_L(\Run{M})$ be a possible
  $L$ view. By the above observation there exists a strategy $\pi_1$
  such that $\vw \in \view_L(\Run{M,\pi_1})$. By $M\in \NDSs$,
  $\view_L(\Run{M,\pi_1})= \view_L(\Run{M,\pi})$, so also $\vw \in
  \view_L(\Run{M,\pi})$, as required. \qed
\end{proof}

This formulation makes
it clear that $H$ cannot communicate any information to $L$ by means of its
strategies.  It is also apparent that allowing $H$ strategies to be
non-deterministic (i.e., functions from $H$ views to a \emph{set} of $H$
actions) would not lead to a different definition of $\tNDS$, since the more
choices $H$ has in a strategy the more $L$-views are compatible with that
strategy.  We remark that in asynchronous systems (in which we use an
asynchronous notion of view), similarly defined notions of non-deducibility on
inputs and non-deducibility on strategies turn out to be
equivalent~\cite{FG95,MZ10}.  The example~\ref{example:nds-ndi} above shows
that this is not the case in synchronous machines, where the two notions are
distinct.

\subsubsection{Unwinding Relations}

Nondeducibility-based definitions of security are quite intuitive, but
they turn out to have some disadvantages as a basis for secure systems
development.  One is that they are not compositional: combining two
systems, each secure according to such a definition, can produce a
compound system that is not secure~\cite{McCullough88}.  For this
reason, some stronger, but less intuitive definitions have been
advocated in the literature.

One of these, McCullough's notion of
restrictiveness~\cite{McCullough88}, is closely related to an approach
to formal proof of systems security based on what are known as
``unwinding relations."  A variety of definitions of unwinding
relations have been proposed in the
literature~\cite{GM84,Rushby92,Mantel01,BFPR03}, in the context of a
number of different underlying systems models and associated
definitions of security for which they are intended to provide a proof
technique.  We propose here a variant of such definitions that is
appropriate to the machine model we consider in this paper, drawing on
definitions proposed by van der Meyden and Zhang \cite{MZ08} for
machines acting under the control of a scheduler.

\begin{definition}[Synchronous Unwinding Relation]\label{def-unwinding}
  A \emph{synchronous unwinding relation} on a system $M$ is a
  symmetric relation $\sim\subseteq S\times S$ satisfying the
  following:
\begin{enumerate}
  \item $s_0\sim s_0$,
  \item $s\sim t$ implies $\obs_L(s)=\obs_L(t)$,
  and
  \item $s\sim t$ implies that for all $a_1,a_2\in A_H$ and $a_3\in A_L$,
    if $s\xrightarrow{(a_1,a_3)}s'$ then there exists a state $t'$
    such that $t\xrightarrow{(a_2,a_3)}t'$, and $s'\sim t'$.
  \end{enumerate}
\end{definition}

Intuitively, an unwinding relation is a bisimulation-like
relation over $S$ that shows $L$ observations are locally uncorrelated with $H$ actions.

\begin{definition}\label{def:res}
  A synchronous machine $M$ satisfies restrictiveness
  ($M\in \tRES$), if there exists a synchronous unwinding relation on $M$.
\end{definition}

Part of the significance of $\tRES$ is that it provides a proof
technique for our notions of nondeducibility, as shown by the
following result, which relates the three notions of security we have
introduced:

\begin{theorem}\label{thm:containments}
  The following containments hold and are strict: $\tRES \subset
  \NDSs\subset \NDIs$.
\end{theorem}
\begin{proof}

  To show that $\tRES \subseteq \NDSs$ we argue as follows.  Suppose
  that $\sim$ is a synchronous unwinding on $M$.  Let $\vw$ be any
  possible $L$ view, and $\pi$ any $H$ strategy.  We have to show that
  $\vw\in \view_L(\Run{M,\pi})$. For this, let $r = s_0 a_1 s_1 \ldots
  a_n s_n$ be a run such that $\vw = \view_L(r)$. We show that there
  exists a run $r' = s'_0 a'_1 s'_1 \ldots a'_n s'_n$ consistent with
  $\pi$ such that $\vw = \view_L(r')$. We proceed inductively, showing
  for each $i= 0\ldots n$ that $\view_L(s_0 a_1 s_1 \ldots a_i s_i) =
  \view_L(s'_0 a'_1 s'_1 \ldots a'_i s'_i)$ and $s_i \sim s_i'$, where
  $s'_0 a'_1 s'_1 \ldots a'_i s'_i$ is a run consistent with $\pi$.
  In the base case, we have $s_0' = s_0$ and the claim is trivial.
  For the inductive case, let $a = \pi(\view_H(s'_0 a'_1 s'_1 \ldots
  a'_i s'_i))$ and $b = \proj_L(a_{i+1})$, and take $a'_{i+1} =
  (a,b)$. Since $\proj_L(a_{i+1}) = \proj_L(a'_{i+1})$ and $\sim $ is
  an unwinding relation, there exists a state $s'_{i+1}$ such that
  $s'_i \xrightarrow{a'_{i+1}} s'_{i+1}$ and $s_{i+1} \sim s'_{i+1}$.
  Further, we conclude that $\obs_L(s_{i+1}) = \obs_L (s'_{i+1})$.
  Since $s'_0 a'_1 s'_1 \ldots a'_i s'_i$ is consistent with $\pi$, so
  is $s'_0 a'_1 s'_1 \ldots a'_i s'_i a'_{i+1}s'_{i+1}$, and
\[\begin{array}{rcl}
\view_L(s'_0 a'_1 s'_1 \ldots a'_i s'_i a'_{i+1}s'_{i+1})
& = & \view_L(s'_0 a'_1 s'_1 \ldots a'_i s'_i) \proj_L(a'_{i+1})\obs_L(s'_{i+1})\\
& = & \view_L(s_0 a_1 s_1 \ldots a_i s_i) \proj_L(a_{i+1})\obs_L(s_{i+1})\\
& = & \view_L(s_0 a_1 s_1 \ldots a_i s_i a_{i+1} s_{i+1})~,
\end{array}\] as required.

Next we show that $\tNDS \subseteq \tNDI$.  Let $\alpha\in A_H^*$ and
$\vw$ be an $L$ observation satisfying $|\alpha|=|\vw|$. We
construct a ``blind'' $H$ strategy $\pi$ as
$\pi(\vw')=\alpha_{|\vw'|+1}$ if $|\vw'|<|\alpha|$ and $\pi(\vw')=a_H$
otherwise, where $a_H$ is an arbitrary action in $A_H$. Since
$M\in\tNDS$, we have $\view_L(\Run{M,\pi})=\view_L(\Run{M})$, so there
exists a run $r\in\Run{M,\pi}$ such that $\view_L(r)=\vw$. By the
construction of $\pi$ we have $Act_H(r)=\alpha$.

That the inclusions are strict follows from
Example~\ref{example:nds-ndi} and Example~\ref{example:res-nds} below.  \qed
\end{proof}

\begin{figure}[t]
\centering \includegraphics[scale=1]{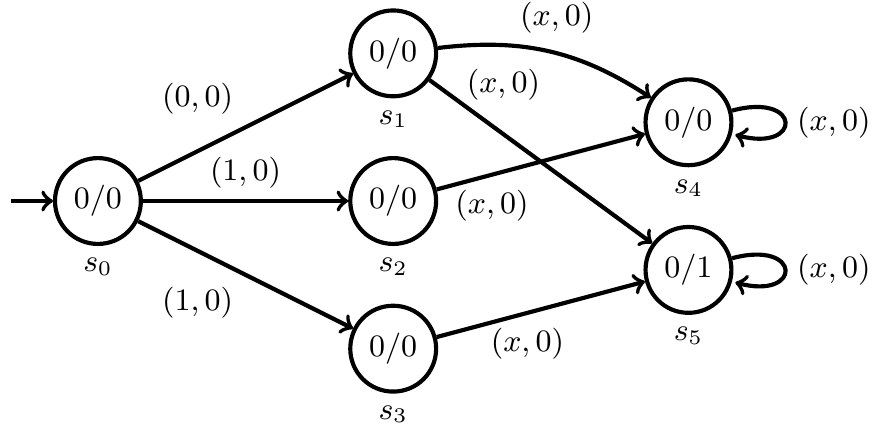}
\caption{A synchronous machine $M$ in $\NDSs$, but not in $\tRES$, where $x\in\set{0,1}$.
Every state $s$ is labelled with a pair $\obs_H(s)/\obs_L(s)$.}\label{fig:res-nds}
\end{figure}

\begin{example}\label{example:res-nds}
  We present a machine in Fig.~\ref{fig:res-nds} that satisfies
  $\tNDS$ but does not satisfy $\tRES$. In this system we let
  $A_H=\set{0,1}$, $A_L=\set{0}$.  We use the conventions from
  Example~\ref{example:nds-ndi}.  One may easily observe that
  the set of $L$ views is given by the regular language
  $000((00)^*+(01)^*)$ and all the views are compatible with every
  possible $H$ strategy. However, there does not exist a synchronous
  unwinding relation. Suppose there were such a relation $\sim$.
  Then $s_0\sim s_0$, and for joint actions
  $(0,0)$ and $(1,0)$, we have $s_0\xrightarrow{(0,0)}s_1$,
  $s_0\xrightarrow{(1,0)}s_2$ and $s_0\xrightarrow{(1,0)}s_3$, and we
  would require $s_1$ to be related to either $s_2$ or $s_3$.
  However, neither $s_2$ nor $s_3$ can be related to $s_1$: from $s_2$
  user $L$ can only observe $(00)^*$ in the future, and from $s_3$
  only $(01)^*$ can be observed by $L$.  Note from $s_1$ both $(00)^*$
  and $(01)^*$ are possible for~$L$.  \qed
\end{example}

In the following sections, we study the complexity of
the notions of security we have defined above.

\section{Synchronous Nondeducibility on Inputs}\label{sec:ndi}

In this section we establish  the following result:
\begin{theorem} For the class of finite state synchronous machines,
$\NDIs$ is \PSPACE-complete with respect to logspace reductions.
\end{theorem}

\subsection{PSPACE-Easiness}

Stating the definition in the negative, a system is not in $\NDIs$ if
there exists an $L$ view $\vw$ and a sequence of $H$ actions $\alpha$
with $|\alpha|=|\vw|$ such that there exists no run $\run$ with
$Act_H(\run)=\alpha$ and $\view_L(\run)=\vw$.  We show that $\NDIs$ is
decidable by a procedure that searches for such an $L$ view $\vw$ and
$H$ action sequence $\alpha$. The key element of the proof is to show
that we need to maintain only a limited amount of information during
this search, so that we can bound the length of the witness
$(\vw,\alpha)$, and the amount of space needed to show that such a
witness~exists.

To show this, suppose we are given a machine $M=\word{S, A, s_0, \rightarrow, O,
  \obs}$.  Given a sequence $\alpha\in A_H^*$ and
a sequence $\vw \in O(A_LO)^*$, we define the set $K(\alpha, \vw)$ to
be the set of all final states of runs $\run$ of $M$
consistent with $\alpha$ and $\vw$,  \ie
such that $Act_H(\run)=\alpha$ and $\view_L(\run)=\vw$.
For each $a\in A_H$, $b\in A_L$ and $o\in O$, we also define the
function $\delta_{a,b,o}:{\cal P}(S)\rightarrow {\cal P}(S)$, by
$$\delta_{a,b,o}(T) = \{t\in S~|~
  \mbox{for some}~t'\in T~\mbox{we have} ~t'\xrightarrow{(a,b)}
  t~\mbox{and} ~\obs_L(t) = o\}~.$$
For the system $M$
define the labelled transition system $LTS(M) = (Q,\Sigma,q_0,\Rightarrow)$ as
follows:
\begin{enumerate}
\item $Q = S\times {\cal P}(S)$,
\item $q_0 = (s_0,\{s_0\})$,
\item $\Sigma=A_H\times A_L\times A_H$,
\item $\Rightarrow \subseteq Q \times \Sigma \times Q$ is the labelled
  transition relation defined by $(s,T) \Rightarrow^{(a,b,a')}
  (s',T')$ if $a\in A_H$, $b\in A_L$, $a'\in A_H$ such that $s
  \xrightarrow{(a,b)} s'$ and
  $T' = \delta_{a',b,\obs_L(s')}(T)$.
\end{enumerate}
Intuitively, the component $s$ in a state $(s,T) \in Q$ is used to
ensure that we generate an $L$ view $\vw$ that is in fact
possible. The components $a,b$ in a transition $(s,T)
\Rightarrow^{(a,b,a')} (s',T')$ represent the actions used to generate
the run underlying $\vw$, and the component $a'$ is used to generate a
sequence $\alpha$. The set $T$ represents $K(\alpha, \vw)$.  More
precisely, we have the following result:

\begin{lemma} \label{lem:ndilts} If $q_0 \Rightarrow^{(a_1,b_1,a'_1)}
  (s_1,T_1)\Rightarrow\dots \Rightarrow^{(a_n,b_n,a'_n)} (s_n,T_n)$,
  then the sequence $\vw = \obs_L(s_0) b_1 \obs_L(s_1)$ $ \ldots b_n \obs_L(s_n)$
  is a possible $L$ view, 
  and $\alpha = a'_1\ldots a'_n$ is a sequence of $H$ actions such
  that $|\vw| = |\alpha|$ and $ K(\alpha, \vw) = T_n$.

  Conversely, for every possible $L$ view $\vw$ with $|\vw|=n$, and
  sequence of $H$ actions $\alpha = a'_1\ldots a'_n$, there exists a
  path $q_0\Rightarrow^{(a_1,b_1,a'_1)} (s_1,T_1)\Rightarrow\dots
  \Rightarrow^{(a_n,b_n,a'_n)} (s_n,T_n)$ such that $\vw = \obs_L(s_0)
  b_1 \obs_L(s_1) \ldots b_n \obs_L(s_n)$ and $ K(\alpha, \vw) = T_n$.
\end{lemma}
\begin{proof}
  We first show that for all $\alpha\in A_H^*$, $v\in O(A_LO)^*$,
  $a\in A_H$, $b\in A_L$ and $o\in O$, we have $K(\alpha a, v b o) =
  \delta_{a,b,o}(K(\alpha, v))$. To show $K(\alpha a, v b o) \subseteq
  \delta_{a,b,o}(K(\alpha, v))$, suppose that $t\in K(\alpha a, v b
  o)$. Then there exists a run $r$ of $M$ such that $\textit{Act}_H(r)
  = \alpha a$ and $\view_L(r) = vbo$ and the final state of $r$ is
  $t$. It follows that $\obs_L(t) = o$.  Thus, we may write $r = r'
  \LT{(a,b)} t$, where $\textit{Act}_H(r') = \alpha$, and $\view_L(r')
  = v$.  Thus, the final state $t'$ of $r'$ is in $K(\alpha,
  v)$. Since $t'\LT{(a,b)}t$ and $\obs_L(t) = o$, it follows that
  $t\in \delta_{a,b,o}(K(\alpha, v))$.

  Conversely, if $t\in \delta_{a,b,o}(K(\alpha, v))$ then by
  definition of $\delta_{a,b,o}$ there exists $t'\in K(\alpha, v)$
  such that $t'\LT{(a,b)} t$ and $\obs_L(t) = o$. By definition of
  $K(\alpha, v)$ there exists a run $r$ of $M$ such that
  $\textit{Act}_H(r) = \alpha$ and $\view_L(r) = v$.  Taking $r'=
  r\LT{(a,b)} t$, we see that $r'$ is a run of $M$ with
  $\textit{Act}_H(r) = \alpha a$ and $\view_L(r') = v b o$. Thus,
  $t\in K(\alpha a, v b o)$, as required.  This completes the proof
  that $K(\alpha a, v b o) = \delta_{a,b,o}(K(\alpha, v))$.

We can now prove the two parts of the result:
\begin{itemize}
\item Suppose $q_0
  \Rightarrow^{(a_1,b_1,a'_1)}(s_1,T_1)\Rightarrow\dots\Rightarrow^{(a_n,b_n,a'_n)}(s_n,T_n)$
  is a run of $LTS(M)$, then by definition
  $r=s_0\LT{(a_1,b_1)}s_1\LT{(a_2,b_2)}\dots\LT{(a_n,b_n)}s_n$ is a
  run of $M$, such that $\view_L(r)=\obs_L(s_0) b_1 \obs_L(s_1) \ldots
  b_n \obs_L(s_n) =\vw$ is a possible $L$ view.
  Moreover, $T_{i+1} = \delta_{a_i,b_i,\obs_L(s_{i+1})}(T_i)$ for $i=
  0\ldots n-1$, where we take $T_0 =\{s_0\}$. Since $T_0 =
  K(\varepsilon, \obs_L(s_0))$, it follows from the above using a
  straightforward induction that $T_n = K(\alpha, v)$, where $\alpha =
  a_1'\ldots a_n'$.

\item Let $\vw$ be a possible $L$ view, then there exists a run
  $r=s_0\LT{(a_1,b_1)}s_1\LT{(a_2,b_2)}\dots\LT{(a_n,b_n)}s_n$ of $M$
  such that $\view_L(r)=\vw$. Given a sequence of $H$ actions $\alpha
  = a'_1\ldots a'_n$,
  we inductively define $T_0 = \{s_0\}$ and $T_{i+1} =
  \delta_{a_i,b_i,\obs_L(s_{i+1})}(T_i)$. It is then immediate by
  definition that we have a path
  $(s_0,\set{s_0})\Rightarrow^{(a_1,b_1,a'_1)}(s_1,T_1)\Rightarrow\dots\Rightarrow^{(a_n,b_n,a'_n)}
  (s_n,T_n)$ in $LTS(M)$. By a straightforward induction using what
  was proved above, we have that $T_n = K(\alpha, \vw)$, where $\alpha
  = a_1'\ldots a_n'$ and $\vw = \obs_L(s_0) b_1 \obs_L(s_1) \ldots b_n
  \obs_L(s_n)$.  \qed
\end{itemize}
\end{proof}
We now note that for an $H$ action sequence $\alpha$ and a possible
$L$ view $\vw$, with $|\vw| = |\alpha|$, there exists no run $r$ such
that $Act_H(\run)=\alpha$ and $\view_L(\run)=\vw$ iff $K(\alpha,\vw) =
\emptyset$.  The existence of such a pair $(\alpha, \vw)$, is
therefore equivalent, by Lemma~\ref{lem:ndilts}, to the existence of a
path in $LTS(M)$ from $q_0$ to a state $(s,T)$ with $T= \emptyset$.
This can be decided in $\NSPACE(O(|M|)) = \DSPACE(O(|M|^2))\subseteq
\PSPACE$.  This proves the following theorem.
\begin{theorem}
$M \in \NDIs$ is decidable in $\PSPACE$.
\end{theorem}

We note, moreover, that since there are at most $|S|\times 2^{|S|}$ states in $Q$,
if there exists a pair $(\alpha, \vw)$ witnessing that $M \not \in \NDIs$
there exists such a pair with $|\alpha| \leq |S|\times 2^{|S|}$.

\subsection{PSPACE-Hardness}

We show that $\NDIs$ is \PSPACE-hard already in the special case of
scheduled machines. The proof is by a polynomial time reduction from
the problem of deciding, given a non-deterministic finite state automaton 
$\Aut$ on alphabet $\Sigma$, if the language $L(\Aut)$ accepted by $\Aut$
is equal to $\Sigma^*$. This \emph{Universality} problem is
PSPACE-hard~\cite{SM73}.

Let $\Aut = \word{Q, Q_0, \Sigma, \delta, F}$ be a non-deterministic
finite state automaton (without $\varepsilon$-transitions), with
states $Q$, initial states $Q_0\subseteq Q$, alphabet $\Sigma$,
transition function $\delta: Q\times \Sigma \rightarrow {\cal P}(Q)$,
and final states $F$. We define $M(\Aut)=\word{S,
  A, s_0, \rightarrow,\obs, O}$ to be a scheduled machine, and use a
function $\sched:S\rightarrow \{H,L\}$ to indicate the agent (if any)
whose actions determine transitions. In view of this, when $\sched (s)
= u$ and $a\in A_u$, we may write $s\xrightarrow{a} t $ to represent
that $s\xrightarrow{b} t$ for all joint actions $b$ with $\proj_u(b)
=a$.  The components of $M(A)$ are defined as follows.
\begin{itemize}
\item $S = Q\cup \set{s_0,s_1,s_2,s_3}$, where $Q\cap\set{s_0,s_1, s_2, s_3}=\emptyset$,
\item $\sched(s_0)=H$ and $\sched(s)=L$ for all $s\in S\setminus\set{s_0}$,
\item $A = A_H\cup A_L$ where $A_L=\Sigma$ and $A_H=\set{h,h'}$,
\item $O=\set{0,1}$,
\item $\obs:\{H,L\}\times S\rightarrow O$ with $\obs_H(s) = 0$ for all $s\in S$
      and $\obs_L(s) = 0$ for all $s\in S\setminus\set{s_2}$, and $\obs_L(s_2) = 1$.
\item $\longrightarrow \subseteq  S\times A\times S$ is defined as
      consisting of the following transitions (using the convention noted above)
      \begin{itemize}
 \item $s_0 \xrightarrow{\ h\ } q$ for all $q\in Q_0$,
 \item $s_0 \xrightarrow{\ h\ } s_2$, provided $Q_0\cap F \neq \emptyset$,
 \item $s_0 \xrightarrow{\ h'\ } s_1$
          and $s_0 \xrightarrow{\ h' \ } s_2$,
      \item $s_1 \xrightarrow{\ a\ } s_1$ and $s_1 \xrightarrow{\ a\ } s_2$ for all $a\in \Sigma$,
      \item $s_2 \xrightarrow{\ a\ } s_2$ for all $a\in \Sigma$,
      \item $s_3 \xrightarrow{\ a\ } s_3$ for all $a\in \Sigma$,
      \item for $q,q' \in Q$ and $a\in A_L = \Sigma$ we have
            $q\xrightarrow{\ a\ } q'$ for all $q'\in \delta(q,a)$
      \item for $q\in Q$ and $a\in A_L= \Sigma$ such that
            $\delta(q,a) \cap F \neq \emptyset$, we have
            $q\xrightarrow{ \ a\ } s_2$,
      \item for $q\in Q$ and $a\in A_L= \Sigma$ such that
            $\delta(q,a)=\emptyset$, we have $q\xrightarrow{\ a\ } s_3$.
      \end{itemize}
\end{itemize}
The construction of $M(\Aut)$ from $\Aut$ can be done in logspace.

Intuitively, the runs of $M(\Aut)$ produce two sets of $L$ views,
depending on whether the first $H$ action is $h$ or $h'$.  In all
circumstances, runs in which $H$ does $h'$ in the first step,
with the first transition to $s_1$, produce an $L$ view for each
sequence in $0\Sigma 1 (\Sigma 1)^*$ or $0\Sigma 0(\Sigma 0)^*(\Sigma
1)^*$ by switching from to $s_2$ at the first occurrence of
observation $1$.  Runs in which $H$ does $h$ in the first step
with a transition to a state in $Q_0$, correspond to simulations of
$\Aut$ and produce two types of $L$ views:
\begin{enumerate}
\item any sequence in $0\Sigma 0(\Sigma 0)^*$
  (by means of a run that stays in $Q$ for as long as possible, and
  moves to $s_3$ whenever an action is not enabled), and
\item any sequence of the form
  $0\Sigma 0 b_1 0 \ldots b_{n-1} 0 b_n 1 (\Sigma 1)^*$ with
  $b_1\ldots b_n \in L(\Aut)$
  (these come from runs that pass through $Q$ and then jump to $s_2$).
\end{enumerate}
  In case $\varepsilon\in L(\Aut)$,
  i.e., $Q_0\cap F \neq \emptyset$, then
  also all sequences in $0\Sigma 1(\Sigma 1)^*$ are produced as the
  $L$ view on a run in which the first transition, with $H$ action
  $h$, is to $s_2$. %

Note that since $L$ is always scheduled after the first step, replacing any
action by $H$ after the first step in a run by any other action of $H$ results
in another run, with no change to the $L$ view. Thus, the only thing that
needs to be checked to determine whether $M(\Aut) \in \tNDI$ is whether the
same can be said for the first step.
Moreover, it can be seen from the above that, for all $\Aut$, and
independently of whether $\varepsilon \in L(\Aut)$, any $L$ view
obtained from a run in which the first $H$ action is $h$ can also be
obtained from a run in which the first $H$ action is $h'$. Thus, to
show $M(\Aut)\in \tNDI$, it suffices to check that any $L$ view
obtained from a run in which the first $H$ action is $h'$ can also be
obtained from a run in which the first action is $h$.

\begin{proposition} 
  $L(\Aut)=\Sigma^*$ iff $M(\Aut)\in\NDIs$.
\end{proposition}
\begin{proof}
For the `only if' part, suppose
  $L(\Aut)=\Sigma^*$. We show that  $M(\Aut)\in \tNDI$.
  As argued above, it suffices to show that any view obtained from a
  run in which the first $H$ action is $h'$ can also be obtained from
  a run in which the first action is $h$.  Let $r= s_0 (h', b_1) t_1 $
  $\ldots (a_n,b_n) t_n$ be a run of $M(\Aut)$, with the $a_i \in A_H$
  and the $b_i \in A_L$. If $t_1 = s_2$ then since $\varepsilon \in
  L(\Aut)$, simply by replacing the first transition by $s_0 \LT{(h,
    b_1)} s_2$ we obtain a run with first $H$ action $h$ that has
  exactly the same $L$ view.  Otherwise $t_1 = s_1$.  If all
  $\obs_L(t_i) = 0$ then we may construct a run of $M(\Aut)$ with the
  same $L$ view as $r$ by taking any transition into $Q$ in the first
  step, then remaining within $Q$ throughout, or making a transition
  to $s_3$ if there is no enabled transition of $\Aut$ on the given
  input $b_i$. Otherwise, let $i$ be the least index with $\obs_L(t_i)
  =1$. Since $L(\Aut) = \Sigma^*$, there exists a run $q_0 \LT{b_2}
  q_1 \ldots \LT{b_i} q_{i-1}$ of $\Aut$ with $q_{i-1}\in F$. By
  construction of $M(\Aut)$, it then follows that
$$r'= s_0 \LT{(h, b_1)} q_0 \LT{(h, b_1)} q_1 \ldots q_{i-2} \LT{(a_i,b_i)} s_2 \ldots \LT{(a_n,b_n)}  s_2$$
is a run with the same $L$ view as $r$ but with first $H$ action $h$.
Thus, $M(\Aut)\in \tNDI$.

For the `if' part, suppose there is a word $w=a_1a_2\dots a_n\not\in L(\Aut)$.
  If $w = \varepsilon$, then the transition $s_0 \LT{h} s_2$ is not
  present in $M(\Aut)$, so there is no run starting with $H$ action
  $h$ that produces the $L$ view $0 b 1$ (with $b\in \Sigma$) that we
  get from $s_0 \LT{(h',b)} s_1$.  Otherwise, for an arbitrary $a_0\in
  \Sigma$, the $L$ view $0a_0 0a_10a_2\dots a_n1$ cannot be obtained
  from runs in which the first $H$ action is $h$, because otherwise
  $w$ would be accepted by $\Aut$. However this view is obtained from
  a run in which the first action is $h'$.  Therefore
  $M(\Aut)\not\in\NDIs$. \qed
\end{proof}

As pointed out at the beginning of this section, this lower bound
result already holds for scheduled machines, and thus $\tNDI$ is
already PSPACE-hard for this subclass.

\section{Nondeducibility on Strategies} \label{sec:nds}

In this section we establish the following theorem:
\begin{theorem}\label{thm-NDS} For the class of finite state synchronous machines,
  and with respect to log-space reductions, $\NDSs$ is \EXPSPACE-complete.
\end{theorem}

\subsection{EXPSPACE-Easiness}

For the proof that \NDSs\ is decidable in \EXPSPACE, we show that the
problem is in \DSPACE$(2^{O(|M|)})$.
It is convenient for this section to consider strategies $\pi$ that
are defined over the larger set $\pvh = O(A_H O)^*$ of candidate views
of $H$, rather than the subset $\view_H(\Run{M})$ of possible views.

 We use the characterization of $\NDSs$ given in
 Proposition~\ref{prop:charnds}.  Let $\pi$ be an $H$
 strategy, and $\beta$ be an $L$ view. 
 Say that $\pi$ {\em excludes} $\beta$ if there does not exist a run
 $r$ consistent with $\pi$ such that $\beta = \view_L(r)$.  Since
 always $\Run{M,\pi} \subseteq \Run{M}$, by
 Proposition~\ref{prop:charnds}, a system $M$ satisfies $\tNDS$ if and
 only if it is not the case that there exists a possible $L$ view
 $\beta$ in $M$ and a strategy $\pi$ such that $\pi$ excludes $\beta$.

Our decidability result and complexity bound is obtained by showing
that if such a strategy exists, then there is one of a particular
normal form, and it can be found using a space-bounded search. The
normal form strategies have a uniform structure, in that the choice of
next action on an $H$ view depends only on the length of the view and
the set of states that $H$ considers possible after that view, given
that $L$'s view $\beta$ has not yet been excluded.  We call this set
of states $H$'s {\em knowledge set}.

More precisely, the knowledge sets are defined as follows.  Given a
candidate $H$ view $\alpha \in O(A_HO)^*$ and a (to be excluded)
candidate $L$ view $\beta \in O(A_LO)^*$ with $|\alpha| \leq |\beta|$,
define $K(\alpha,\pi,\beta)$ to be the set of all final states of runs
$r$ consistent with $\pi$ such that $\view_H(r) = \alpha$ and
$\view_L(r)$ is a prefix of $\beta$.

These knowledge sets can be obtained in an incremental way using the
update operators
$\delta_{a_H,o_H,a_L,o_L}:\power{S} \rightarrow \power{S}$ defined for
each $a_H\in A_H$, $a_L\in A_L$, and $o_H, o_L\in O$, to map $T\in
\power{S}$ to
\[ 
\delta_{a_H,o_H,a_L,o_L}(T) = \{ s \in S \ | \
         \exists t \in T \text{ with }
    t \LT{(a_H,a_L)} s \text{ and } \obs_H(s) = o_H \text{ and } \obs_L(s) = o_L\}\mathpunct.
\]
The incremental characterisation is given in the following lemma.

\begin{lemma} \label{lem:delta} Suppose that $\pi(\alpha) = a_H$ and
  $|\alpha| = |\beta|$. Then $
  \delta_{a_H,o_H,a_L,o_L}(K(\alpha,\pi,\beta)) =
    K(\alpha a_H o_H,\pi, \beta a_Lo_L)$.
 \end{lemma}

\begin{proof}
  We first show that $$ \delta_{a_H,o_H,a_L,o_L}(K(\alpha,\pi,\beta))
  \subseteq K(\alpha a_H o_H, \pi, \beta a_Lo_L)\mathpunct.$$  Suppose
  $t\in
  \delta_{a_H,o_H,a_L,o_L}(K(\alpha,\pi,\beta))$. We show that $t\in
  K(\alpha a_Ho_H, \pi,\beta a_Lo_L)$.  We have that there exists
  $s\in K(\alpha,\pi,\beta)$ such that
  $s\xrightarrow{(a_H,a_L)} t$ and $\obs_H(t)= o_H$ and $\obs_L(t) =o_L$.
  Thus there exists a run $r$, consistent with $\pi$, and with
  final state $s$, such that $\view_H(r) = \alpha$ and $\view_L(r) =
  \beta$. Since $\pi(\alpha) = a_H$, we obtain that the run
  $r\, (a_H,a_L)\, t$ is consistent with $\pi$ and justifies $t\in
  K(\alpha a_Ho_H, \pi,\beta a_L o_L)$.

  Conversely, suppose $t\in K(\alpha a_Ho_H, \pi,\beta a_L o_L)$.
  Then there exists a run $r'$ consistent with $\pi$, which can be
  written in the form
  $r\, (a_H,a_L)\, t$ with $\view_H(r) = \alpha$, and $\obs_H(t)=o_H$,
  and $\view_L(r) = \beta$ and $\obs_L(t)=o_L$. Let $s$ be the final
  state of $r$.  Then we have $s\in K(\alpha,\pi,\beta)$.  It is now
  immediate that $t\in
  \delta_{a_H,o_H,a_L,o_L}(K(\alpha,\pi,\beta))$. \qed
\end{proof}

The following result shows that it suffices to consider strategies in
which the choice of action depends only on the time and $H$'s
knowledge set, given the $L$ view being excluded.

\begin{lemma}\label{lem:normalform}
  If there exists an $H$ strategy $\pi$ that excludes $\beta$, then there exists
  an $H$ strategy $\pi'$ that also excludes $\beta$, and has the property that
  for all $H$ views $\alpha$ and~$\alpha'$, if  $K(\alpha,\pi',\beta) = K(\alpha',\pi',\beta)$ and  $|\alpha| = |\alpha'|$
  then $\pi'(\alpha) = \pi'(\alpha')$.
\end{lemma}

\begin{proof}
  Suppose that $\pi$ excludes $\beta$.  For purposes of the proof,
  note that we can assume without loss of generality that $\beta$ is
  infinite --- this helps to avoid mention of views longer than
  $\beta$ as a separate case. (Note that it is equivalent to say that $\pi$
   excludes some prefix of $\beta$.)

  Let $f$ be any mapping from $\pvh$ to $\pvh$ such that for all
  $\alpha,\alpha'\in \pvh$ we have
  \begin{enumerate}
    \item $|f(\alpha)| = |\alpha|$,
      \item  $K(\alpha, \pi, \beta) =
    K(f(\alpha), \pi, \beta)$,
  \item if $|\alpha| =|\alpha'|$ and $K(\alpha, \pi, \beta) =
    K(\alpha', \pi, \beta)$, then $f(\alpha) = f(\alpha')$.
  \end{enumerate}
  Such a mapping always exists; intuitively, it merely picks, at each
  length, a representative $f(\alpha) \in [\alpha]_{\sim}$ of the
  equivalence classes of the equivalence relation defined by
  $\alpha\sim \alpha'$ if $|\alpha|=|\alpha'|$ and $K(\alpha, \pi,
  \beta) = K(\alpha', \pi, \beta)$.

  Now define the mapping $g$ on $\pvh$ as follows.  Let $\alpha_0 =
  \obs_H(s_0)$ be the only possible $H$ view of length $0$.  For
  $\alpha \in \pvh$ of length $0$, we define
  $g(\alpha)=\alpha$.
  For longer $\alpha$, we define $g(\alpha a o) = f(g(\alpha))
  \pi(f(g(\alpha)) o$.  Also, define the strategy $\pi'$ by
  $\pi'(\alpha) = \pi(f(g(\alpha)))$.

  We claim that for all $\alpha \in \pvh$ we have $K(\alpha, \pi',
  \beta) = K(g(\alpha), \pi, \beta)$.  The proof is by induction on
  the length of $\alpha$.  The base case is straightforward, since
  $\alpha_0$ is consistent with all strategies, so $K(\alpha_0, \pi',
  \beta) = \{s_0\} = K(\alpha_0, \pi, \beta)$, and $K(\alpha, \pi',
  \beta) = \emptyset = K(\alpha, \pi, \beta)$ for $\alpha \neq
  \alpha_0$ with $|\alpha| = 0$.
  Suppose the claim holds for $\alpha \in \pvh$ of length
  $i$.  Let $\alpha a o \in \pvh$.  By induction and (2), $K(\alpha,
  \pi', \beta) = K(g(\alpha), \pi, \beta)= K(f(g(\alpha)), \pi,
  \beta)$.
  Let $\beta' a_L o_L$ be the prefix of $\beta$ of length $|\alpha|+1$.
  Since action $a = \pi'(\alpha) = \pi(f(g(\alpha))$,
  using Lemma~\ref{lem:delta}, we have that
  $$\begin{array}{rcl}
  K(\alpha a o, \pi'\!, \beta) & = & K(\alpha a o , \pi'\!, \beta' a_L o_L)\\
  & = & \delta_{a,o,a_L, o_L}(K(\alpha, \pi', \beta'))\\
  & = & \delta_{a,o,a_L,o_L}(K(f(g(\alpha)), \pi, \beta'))\\
  & = & K(f(g(\alpha)) a o , \pi, \beta' a_L o_L)\\
  & = & K(f(g(\alpha)) a o , \pi, \beta) \\
    & = & K(g(\alpha a o) , \pi, \beta) ~~.\\
  \end{array}$$

  To see that $\pi'$ has the required property, if $K(\alpha,
  \pi',\beta) = K(\alpha',\pi',\beta)$ with $|\alpha|=|\alpha'|$, then
  we have $K(g(\alpha),\pi,\beta)=K(g(\alpha'),\pi,\beta)$.  By (3) we
  have $f(g(\alpha))=f(g(\alpha'))$. Therefore
  $\pi'(\alpha)=\pi(f(g(\alpha)))=\pi(f(g(\alpha')))=\pi'(\alpha')$, by
  definition.

  Since $\pi$ excludes $\beta$, there exists a length $n$ such that
  for all $\alpha \in \pvh$ with $|\alpha|=n$, %
  we have $K(\alpha, \pi,\beta)= \emptyset$. Thus, we also have for
  all $\alpha$ of length $n$ that $K(\alpha, \pi',\beta)= K(g(\alpha),
  \pi, \beta) = \emptyset$. This means that $\pi'$ also excludes
  $\beta$.  \qed
\end{proof}

\smallskip

\newcommand{\K}{{\cal K}}

Based on Lemma~\ref{lem:normalform}, we construct a transition system
$T(M)=(Q,q_0\Rightarrow)$ that simultaneously searches for the
strategy $\pi$ and an $L$ view $\beta$ that is
excluded by $\pi$. The components are defined by:
\begin{enumerate}
 \item $Q=\power{S}\times \power{\power{S}}$,
 \item $q_0=(\set{s_0},\set{\set{s_0}})$,
 \item the transition relation $\Rightarrow$ is defined by $(U,\K)
  \Rightarrow^{(\rho,a_L,o_L)} (U',\K')$ if
 \begin{enumerate}
 \item  $\rho: \K\rightarrow A_H$, and $a_L\in A_L$ and $o_L\in O$,
 \item $U' = \{t~|~$ there exists $s\in U$, and a transition
   $s\LT{(a'_H,a_L)}t$, with $a'_H \in A_H$ and $o_L = \obs_L(t)~\}
   \neq \emptyset$, and
  \item $\K' = \{ \delta_{a_H,o_H,a_L,o_L}(k) ~|~k\in \K \mbox{ and } a_H=\rho(k) \mbox{ and } o_H\in O \}$.
   \end{enumerate}
\end{enumerate}
Intuitively, the component $U$ in a state $(U,\K)$
is used to ensure that the view $\beta$ that we construct is in fact possible in $M$.
The component $\K$
represents a collection of all possible knowledge sets that $H$ can be
in at a certain point of time, while attempting to  exclude
$\beta$.  More specifically, each set $k$ in $\K$
corresponds to $\alpha\in \pvh $ such that $k= K(\alpha,\pi,\beta)$.  In a
transition, we both determine the next phase of $\pi$, by extending
$\pi$ so that $\pi(\alpha) = \rho(K(\alpha, \pi,\beta))$, and extend
$\beta$ to $\beta a_L o_L$. 
Moreover, an $H$ strategy generated by $T(M)$ is only sensitive to
$H$'s knowledge set and lengths of runs,
i.e., it satisfies that $K(\alpha,\pi',\beta) = K(\alpha',\pi',\beta)$ and $|\alpha| = |\alpha'|$ implies
$\pi'(\alpha) = \pi'(\alpha')$. In the above construct, $\rho$ represents the local choice of $H$ that
depends only on the knowledge set of $H$.

The following result justifies
the correspondence between the
transition system $T(M)$ and $\tNDS$.

\begin{lemma} \label{lem:ndsreduction} A machine $M$ does not satisfy
  $\tNDS$ iff $T(M)$ contains a path $q_0 \Rightarrow^* (U,\K)$ to a
  state where $\K = \emptyset$.
\end{lemma}

\begin{proof}
  We first prove the implication from left to right.  Suppose first
  that $M$ does not satisfy $\tNDS$, witnessed by the fact that $\pi$
  excludes the possible $L$ view $\beta= o_0 b_1o_1 b_2, \ldots b_n
  o_n$. We may assume without loss of generality that no strict prefix
  of $\beta$ is excluded.  By Lemma~\ref{lem:normalform}, we may
  assume that $\pi$ has the property that it takes the same value on
  $H$ views $\alpha,\alpha'$ that have the same length and have
  $K(\alpha, \pi,\beta) = K(\alpha', \pi,\beta)$. We construct a path
  \[ q_0 = (U_0, \K_0) \Rightarrow^{(\rho_1, b_1, o_1)} (U_1,\K_1)
  \Rightarrow^{(\rho_2, b_2, o_2)} \ldots \Rightarrow^{(\rho_n, b_n,
    o_n)} (U_n,\K_n)\] in the transition system $T(M)$, by defining
  the functions $\rho_i:\K_{i-1} \rightarrow A_H$ for $i \geq 1$,
  and then deriving $U_i$ from $U_{i-1}$ using the equation in clause
  3(b) of the definition of $T(M)$, and deriving $\K_i$ from
  $\K_{i-1}$ and $\rho_i$ using the equation in clause 3(c).  (This
  guarantees that each step satisfies all the conditions of the
  definition of $\Rightarrow$, except the requirement in 3(b) that
  $U'\neq \emptyset$; we check this below.)  The construction will
  have the property that every $k\in \K_{i-1}$ is equal to some
  $K(\alpha, \pi, \beta)$ with $\alpha\in \pvh$ of length $i-1$.  This
  means that we may define $\rho_i (k) = \pi(\alpha)$.  Note that
  $\rho_i$ is well-defined, by the assumption on $\pi$.  More
  precisely, we claim that for each $i=0\ldots n$, $\K_i$ is a subset
  of the set $\{K(\alpha,\pi,o_0b_1o_1 \ldots b_i o_i)~|~
  |\alpha|=i\}$.  Note that this means that if $k\in \K_n$ then $k =
  \emptyset$, for else we have an $H$ view $\alpha$ of length
  $|\beta|$ such that $K(\alpha,\pi,\beta)\neq \emptyset$, which
  implies that $\pi$ does not exclude $\beta$.  Thus $\K_n =
  \emptyset$, as required for the right hand side of the result.

The proof of the claim is by induction on $i$. The base case of $n=0$ is
immediate from that fact that $\beta$ is a possible view, so $K(\obs_H(s_0), \pi,o_0) = \set{s_0}$.
Suppose $k' \in \K_{i+1}$. We show $k' = K(\alpha',\pi,o_0b_1o_1 \ldots b_i o_i b_{i+1} o_{i+1})$
for some $\alpha'$ of length $i+1$.
By definition of  $\K_{i+1}$,  there exists $k\in \K_i$,
and $o_H\in O$, such that with $a_H = \rho(k)$, we have
$k' = \delta_{b_{i+1},o_{i+1},a_H,o_H}(k)$.
By induction, there exists  $\alpha \in \pvh$ of length $i$ such that
$k= K(\alpha,\pi,o_0b_1o_1 \ldots b_i o_i)$.
By Lemma~\ref{lem:delta}, it follows that $k' = K(\alpha a_H o_H,\pi,o_0b_1o_1 \ldots b_i o_i b_{i+1} o_{i+1})$,
as required.

It remains to show that $U_i\neq \emptyset$ for each $i= 1\ldots n$. For this, note that since $\beta$ is a possible $L$ view, there exists a
run $s_0 \LT{(a_1,b_1)} s_1  \LT{(a_2,b_2)} \ldots  \LT{(a_n,b_n)} s_n$ such that
$\obs_L(s_i) = o_i$ for $i = 1\ldots n$. A straightforward induction shows that for each $i$, we have $s_i\in U_i$,
so in fact $U_i\neq \emptyset$, as required.

For the other direction, suppose that
\[ q_0 = (U_0,\K_0) \Rightarrow^{(\rho_1, b_1, o_1)} (U_1, \K_1)
\Rightarrow^{(\rho_2, b_2, o_2)} \ldots \Rightarrow^{(\rho_n, b_n,
  o_n)} (U_n,\K_n)\] and $\K_n = \{\emptyset\}$.  We construct a
strategy $\pi$ that excludes $\beta = \obs_L(s_0) b_1 o_1\ldots b_n o_n$.  A
straightforward induction using clause 3(b) of the definition of $T(M)$ shows
that $\beta$ is a possible $L$ view in $M$.  The construction of $\pi$ is done
inductively, by defining a sequence of strategies $\pi_0, \pi_1, \ldots,
\pi_{n}$ such that if $i\leq j$ then $\pi_i$ and $\pi_j$ agree on all $H$
views of length at most $i-1$. At each stage of the construction, we claim
that for all $H$ views $\alpha$ of length $i\geq 0$, if $K(\alpha, \pi_i,
\beta) \neq \emptyset$ then $K(\alpha, \pi_i, \beta)\in \K_i$.  Inductively,
we let $\pi_0$ be any strategy and define $\pi_{i+1}(\alpha) =
\rho_{i+1}(K(\alpha,\pi_i,\beta))$ if $|\alpha|=i$ and
$K(\alpha,\pi_i,\beta)\neq \emptyset$, and $\pi_{i+1}(\alpha) = \pi_i(\alpha)$
otherwise.  Evidently, $\pi_{i+1}$ is well defined by the claim that
$K(\alpha, \pi_i, \beta) \neq \emptyset$ then $K(\alpha, \pi_i, \beta)\in
\K_i$.  Also this definition plainly satisfies the condition that if $i\leq j$
then $\pi_i$ and $\pi_j$ agree on all views of length at most $i-1$.  Note
also that since $K_n = \{\emptyset\}$, by the claim there does not exist an
$H$ view $\alpha$ of length $n$ such that $K(\alpha, \pi_n, \beta) \neq
\emptyset$. It follows that $\pi_n$ excludes $\beta$.

It therefore suffices to show that the definition satisfies the claim.
Note that it holds trivially for any strategy if $i=0$. Suppose that  for all $H$ views $\alpha$ of length $i$,
if  $K(\alpha, \pi_i, \beta) \neq \emptyset$ then  $K(\alpha, \pi_i, \beta)\in \K_i$.
Let $\alpha a_H o_H$ be an $H$ view of length $i+1\leq n-1$
with $K(\alpha a_H o_H, \pi_{i+1}, \beta) \neq \emptyset$. Then also
 $K(\alpha, \pi_{i+1}, \beta) \neq \emptyset$ and $\pi_{i+1}(\alpha) = a_H$. Since $\pi_i$ and $\pi_{i+1}$ agree
 on views of length at most $i-1$, we also have  $K(\alpha, \pi_{i+1}, \beta)=K(\alpha, \pi_{i}, \beta) \neq \emptyset$,
 so  $K(\alpha, \pi_{i+1}, \beta)\in \K_i$.
 By Lemma~\ref{lem:delta},  we have that
 $ K(\alpha a_H o_H, \pi_{i+1}, \beta)= \delta_{b_{i+1},o_{i+1},a_H,o_H}$ $(K(\alpha, \pi_{i+1}, \beta)) \in \K_{i+1}$,
 as required. \qed
 \end{proof}

\medskip

We obtain the claimed complexity bound from
Lemma~\ref{lem:ndsreduction}, simply by noting that it reduces $\tNDS$
to a reachability problem in the transition system $T(M)$.  Since the
states of the system $T(M)$ can be represented in space $O(|S|\cdot
2^{|S|}) = 2^{O(|S|)}$, we obtain from Savitch's theorem that we can do
the search in DSPACE$(2^{O(|S|)})$.

\subsection{EXPSPACE-Hardness}

To show that $\NDSs$ is \EXPSPACE-hard, we show how to encode the game
BLIND-PEEK of Reif~\cite{Reif84}. We need only scheduled machines for
the encoding, so the problem is \EXPSPACE-hard already for this
subclass.

\subsubsection{The Game BLIND-PEEK}

BLIND-PEEK is a variant of the two-player game PEEK introduced by
Stockmeyer and Chandra~\cite{SC79}. A PEEK game consists of a box with
two open sides that contains horizontally stacked plates; the players
sit at opposite sides of the box.  Each plate has two positions, `in'
and `out', and contains a knob at one side of the box, so that this
plate can be controlled by one of the players.  At each step, one of
the two players may grasp a knob from his side and push it `in' or
`out'. The player may also pass. Both the top of the box and the
plates have holes in various positions, and each hole is associated to
a player. If, just after a move of player $a\in\set{1,2}$, the plates are
positioned so that for one of the player's holes in the top of the box,
it is possible to peek through from the top of the box to the bottom 
(i.e., each plate has a hole positioned directly underneath the top
hole), then player $a$ wins.  In PEEK, both players can observe the
position of all plates at all times. BLIND-PEEK~\cite{Reif84} (more
formally, the game $G^{2B}$ of that paper) is a modification of PEEK
in which player $1$'s side of the box is partially covered, so that it
is not possible for player $1$ to see the positions of the plates
controlled by player $2$. We may represent the game formally as follows:

\def\pos{\mbox{\tt Pos}}
\def\pass{\mbox{\tt Pass}}
\def\Move{\mbox{\tt Move}}
\def\mmove{\mbox{\tt move}}

\newcommand{\winf}{\Phi}
\newcommand{\plate}{P}

\begin{definition}
An instance $G$ of the BLIND-PEEK game is given by a tuple
$(n, n_1, \winf_1, \winf_2, \nu_0)$ where $n$ and $n_1$ are natural numbers with $n_1<n$,
$$\winf_1 = \bigvee_{j=1}^{h_1}  \gamma^1_j~~ \text{and}~~ \winf_2= \bigvee_{j=1}^{h_2} \gamma^2_j $$
are disjunctive normal form formulas over the
set of atomic propositions $\{\plate_1, \ldots, \plate_n\}$, and $\nu_0: [1..n]\rightarrow \{0,1\}$
represents a boolean assignment to these propositions. The size of the instance is $O(n(h_1 + h_2))$.
\end{definition}

Here $h_i$, for  $i \in \{1,2\}$, is
the number of holes on the top of the box for
each player.  Intuitively, $n$ gives the number of plates, and the
propositions $P_k$ for $1\leq k \leq n$ correspond to the positions of
the plates, which can be either $in$ ($P_k$ false) or $out$ ($P_k$
true), and a state of the game $G$ is given by a mapping $\nu : [1..n]
\rightarrow \{0,1\}$, with $P_k$ true at $\nu$ just when $\nu(k) =1$.
The total number of states in a game is thus exponential in the size
of the game. The assignment $\nu_0$ specifies the initial state of the game.  
The number $n_1$ specifies the number of plates associated to player 1; we take these to
be plates $1..n_1$. The formula $\winf_i$ gives the winning condition for
player $i$.  Each disjunct $\gamma^i_j$ corresponds to one of the holes on the
top of the box that is associated to player $i$. Which literals are in
$\gamma^i_j$ depends on how the hole in the top of the box aligns with a hole
on the plates when these are in or out. If there is always an alignment with a
hole on plate $k$ then $\gamma^i_j$ contains neither $P_k$ nor $\neg P_k$. If
there is an alignment only when the $k$-th plate is $out$ then $\gamma^i_j$
contains the literal $P_k$, and conversely, if there is an alignment only when
the $k$-th plate is $in$ then $\gamma^i_j$ contains the literal $\neg P_k$.
(If there is never an alignment then we may include both $P_k$ and $\neg P_k$,
but, obviously, we may just as well remove the hole from the game.)

Players~1 and~2 play in turn by moving one of their plates or
passing. As the players' plate numbers partition the set $[1\dots n]$,
we can denote the moves of the players $\mmove_i$ with $1 \leq i \leq
n$; if $i \in [1..n_1]$ it is
a move of player~1, and a move of player~2 otherwise.  We let
$\Move_1=\{\mmove_i~|~1 \leq i \leq n_1\} \cup \{ \pass \}$ and
$\Move_2=\{\mmove_i~|~ n_1+1 \leq i \leq n\} \cup \{ \pass \}$.

A \emph{play} $\varrho$ in $G$ is an alternating sequence  of player~1 and
player~2 moves of the form
\[
\varrho = \nu_0 \xrightarrow{\ \lambda_1\ }   \nu_1 \xrightarrow{\ \lambda_2\ }   \nu_2
\quad \cdots \quad
\nu_{i-1}
\xrightarrow{\ \lambda_i\ }   \nu_i
\]
where $\lambda_{j}$ is a player~1 move in $\Move_1$ when $j$ is odd and a  player~2 move in $\Move_2$ when $j$ is even.
Moreover, if
$\lambda_l=\pass$ then $\nu_{l}=\nu_{l-1}$ and  if
$\lambda_l=\mmove_k$, then $\nu_l(k)=1-\nu_{l-1}(k)$ and
$\nu_l(j)=\nu_{l-1}(j)$ for $j \neq k$.

A state $\nu$ is  winning for player $p$ if it satisfies the formula $\winf_p$.
A play $\varrho$ is {\em winning for player $p$} if  it contains a state $\nu_k$
immediately after a move by player $p$ that is winning for that player,
and there is no earlier such winning  state for the other player.
Otherwise the play is undecided.

We are interested in the problem of deciding whether there is a
winning strategy for player 1, i.e., a way for the player to choose
their moves that guarantees, whatever the other player does, that
player 1 will win. Strategies usually choose a next move based on what
the player has been able to observe over a play of the game.  In the
case of the game $G$, the information directly visible to player 1 in
a state is just the position of plates $1..n_1$.  Player 2 sees the
position of all plates. A player also remembers the sequence of moves
they have played at their turn.  At each step of the play, the player
is also advised whether any player has won the game.  (In a physical
realization of the game, if a player were to peek through a hole they
would be able to see which is the topmost plate that blocks it. We
assume that the player does {\em not} get this information in the
formal game.  One can imagine a physical realization in which a
referee peeks through the holes and announces the result.)

Since the winning condition is a discrete, state-based condition,
deterministic strategies suffice.  Moreover, note that, except for the
information about who has won the game, the effect of the players'
moves on the information directly visible to player 1 is
deterministic: we can deduce the position of player 1's plates from
the moves that player 1 has made so far in the game.  Thus, every
undecided play in which player 1 has made a particular sequence of
moves yields the same view for player 1, and on a deterministic
strategy, player 1 must make the same next move on all such
plays. This means that we may represent a player 1 strategy simply by
a (finite or infinite) sequence of moves $\Lambda = \lambda_1,
\lambda_3, \lambda_5 \ldots$.  Such {\em a strategy is winning for
  player 1} if every play with this sequence of moves by player 1 is
winning for player 1.

The following result characterizes the complexity of BLIND-PEEK.

\begin{theorem} {\bf \cite{Reif84}}
The game  BLIND-PEEK is complete for EXPSPACE under log-space reductions.
\end{theorem}

\def\play{\text{Play}}

We use this result to show that \tNDS\ is EXPSPACE-hard.
Given an instance $G$ of BLIND-PEEK, we construct a  synchronous system
$M(G)$ of size polynomial in the size of $G$,
with the following property: player 1 has a winning strategy
 in $G$ iff there exists an $L$ view $\vw_L$ and
an $H$ strategy $\pi$ that excludes $\vw_L$ in $M(G)$.

Intuitively, the winning strategy $\Lambda$ of player 1 in the game $G$ will be encoded
within the sequence of $L$ actions contained in the view $\vw_L$. The role of the
$H$ strategy $\pi$ in the machine $M(G)$ is  to help in the verification that the
strategy $\Lambda$ is winning, by ensuring that the view $\vw_L$ cannot occur when this is the case.
As we cannot encode the exponential number of possible states of the BLIND-PEEK
game $G$ directly in the polynomial number of  states of $M(G)$,
we use an encoding trick, which is to represent the state of the game
as the set of states of $M(G)$ that are consistent with the
$H$ view. Roughly, each such consistent state corresponds to one
of the plates; there are some additional states for initialization and book-keeping
related to the winning condition.

\subsubsection{High level structure of $M(G)$.}

We let $n_2 = n- n_1$ be the number of plates that can be moved by player 2.

\def\cw{\mbox{\tt checkwin}}
\def\isblocked{\mbox{\tt isBlocking}}
\def\isopen{\mbox{\tt isOpen}}
The machine $M(G)$ will be a scheduled machine (see Section~\ref{sec:schedmachine}),
with  deterministic schedule following the regular expression
$\bot(LHL\bot H^{h_2})^\omega$.
Here occurrences of $H$ and $L$ indicate which agent's action
the transition is allowed to depend upon, and
$\bot$ is for a system step that is independent
of  both agents $H$ and $L$. We call each instance of the infinitely repeated block 
$LHL\bot H^{h_2}$ a {\em round}, and use indices, as in
$L_1H_0L_2\bot H_1..H_{h_2}$ to refer to the stages of the round.

The alphabet of actions for $L$ and $H$ are $A_L$ and $A_H$
defined by:
\begin{eqnarray*}
  A_L^- & = & \{ \mmove_i \ | \  i \in [1\dots n_1] \} \\
  A_L & = & A_L^- \cup \{ \cw \} \\
  A_H & = & \{ \isopen_k \ | \ k \in [1..h_1] \} \cup \{ \isblocked_k \ | \ k \in [1
  \dots n] \} \mathpunct.
\end{eqnarray*}

Informally, the behaviour of $M(G)$ in each step is as follows:

\begin{enumerate}
\item In the first $\bot$ step of the schedule, the machine nondeterministically makes a transition to one of
      $n$ subsystems, each of which monitors one particular plate of the game. Neither $H$
      nor $L$ is able to see which subsystem they are actually in during subsequent transitions.
      The machine then moves into the cyclically repeated rounds.
\item The $L_1$ stage of each round allows $L$ to perform one move of player 1 in
      the game $G$ (using an action $\mmove_i$) or to pass (using the action $\cw$).
      This stage corresponds to a move according to a blindfold strategy of player~1.
\item In the following $H_0$ stage, agent $H$ is given an opportunity to
      assert that the last $L$ move has achieved a win of player~1,
      by specifying a hole $j$ of player 1 and claiming that it is possible to peek through
      (using an action $\isopen_j$). $H$ may also pass (using any other action).
\item The following $L_2$ stage allows $L$ to check, by performing a ``\cw" action,
      if a winning state has been reached, as claimed by $H$ sometime before. If
      it is so, then some of the $L$ views will be ruled out.  If it
      is still not a win, or $H$ has not yet asserted a win, or $H$
      has made a mistake, this ``\cw'' action will not rule out
      any $L$ view. (Once a ``\cw'' fails in this way, no $L$ views will be
      excluded thereafter.)
\item The next stage $\bot$ simulates a move of player~2 in the game $G$.

\item During the last $h_2$ stages, agent $H$ is given an opportunity to assert that
      the last player~2 move is not a win, and explain this claim by pointing,
      for each of the $h_2$ peek-holes of
      player~2, to a plate $j$ that blocks that hole,
      using an action $\isblocked_j$.
\end{enumerate}
The observations of the agents $L,H$ in $M(G)$ are
defined so that neither agent ever learns which plate is being
simulated on the current run. Agent $L$ observes the player 1 moves and a
result of any ``\cw" actions. Agent $H$ observes all  moves by either player.

Intuitively, suppose player~1 has a winning strategy and agent $L$
faithfully follows this strategy in the $L_1$ stage of each round.
Suppose also that agent $H$, who knows every previous move of the play, always
makes correct assertions about whether holes are open or blocked.
Then $L$ is guaranteed to be able to eventually make a successful ``checkwin'' 
and get some $L$ views ruled out. To handle the case where $H$ makes incorrect assertions,
the construction ensures that no $L$ view is eliminated if this
happens. Thus, the statement that there is some  $H$ strategy that eliminates
an $L$ view corresponds to the statement that $H$ has a way of making
correct assertions in order to show that the player 1 strategy is winning.

\def\mmoved{\mmove}
\def\lastmove{\mathbb{M}}
\newcommand{\win}{\mathit{win}}
\newcommand{\lose}{\mathit{error}}
\newcommand{\C}{\mathbb{C}}
\newcommand{\B}{\mathbb{B}}
\newcommand{\Plates}{\mathbb{P}}
\newcommand{\F}{\mathbb{F}}

\subsubsection{States and observations of $M(G)$.}

Formally, the state space of $M(G)$ is defined as $$\set{s_0}\cup (\C\times
\Plates\times \B\times \lastmove) \cup (\C\times \F) \mathpunct,$$
where $s_0$ is the initial state, and the components are as follows:
\begin{itemize}
\item $\C=\set{ L_1, H_0, L_2, \bot, H_1, \cdots, H_{h_2}}$ encodes
a clock that represents the current stage in a round of the cyclic part of the schedule,
\item $\Plates=\set{1,\dots n}$ represents the plate being monitored,
\item $\B=\set{0,1}$ encodes whether the current plate is ``in'' or ``out'',
\item $\M=\Move_1\cup \Move_2\cup\set{\bot}$ records the most recent move in the play
(and $\bot$ for simulation steps that do not correspond to game steps), and
\item $\F = \set{(\win,\bot),(\lose,\bot),(\win,1),(\lose,1),(\lose,2)}$  records the
result of claims made by $H$.
\end{itemize}

\newcommand{\doneobs}{\mathtt{end}}

The observation mappings for $H$ and $L$ on these  states are given by: 
\begin{itemize}
\item $\obs_L((c, r,1))=1$ and $\obs_L((c, r,2))=2$, where $c\in\C$ and $r\in\set{\win, \lose}$,
      and $\obs_L(s)=\bot$ for all $s$ not of this form.
\item $\obs(s_0) = \bot$, and $\obs_H((c,i,k,a))=a$ and $\obs_H(s)=\doneobs$ for all $s\in \C\times \F$.
\end{itemize}

States of $M(G)$ of type $(c,i,k,a) \in \C\times \Plates\times \B\times \lastmove$
encode information about the effect of the play of the game so far on a particular plate:  $c
\in \C$ indicates the current stage of the simulation, $i \in [1..n]$
is a \emph{monitored} plate, $k \in \B$ is the position of the plate
$i$ and $a \in \lastmove$ is the most recent move made in the game,
or $\bot$ if none. 

\begin{figure}
\centerline{\includegraphics[height=6cm]{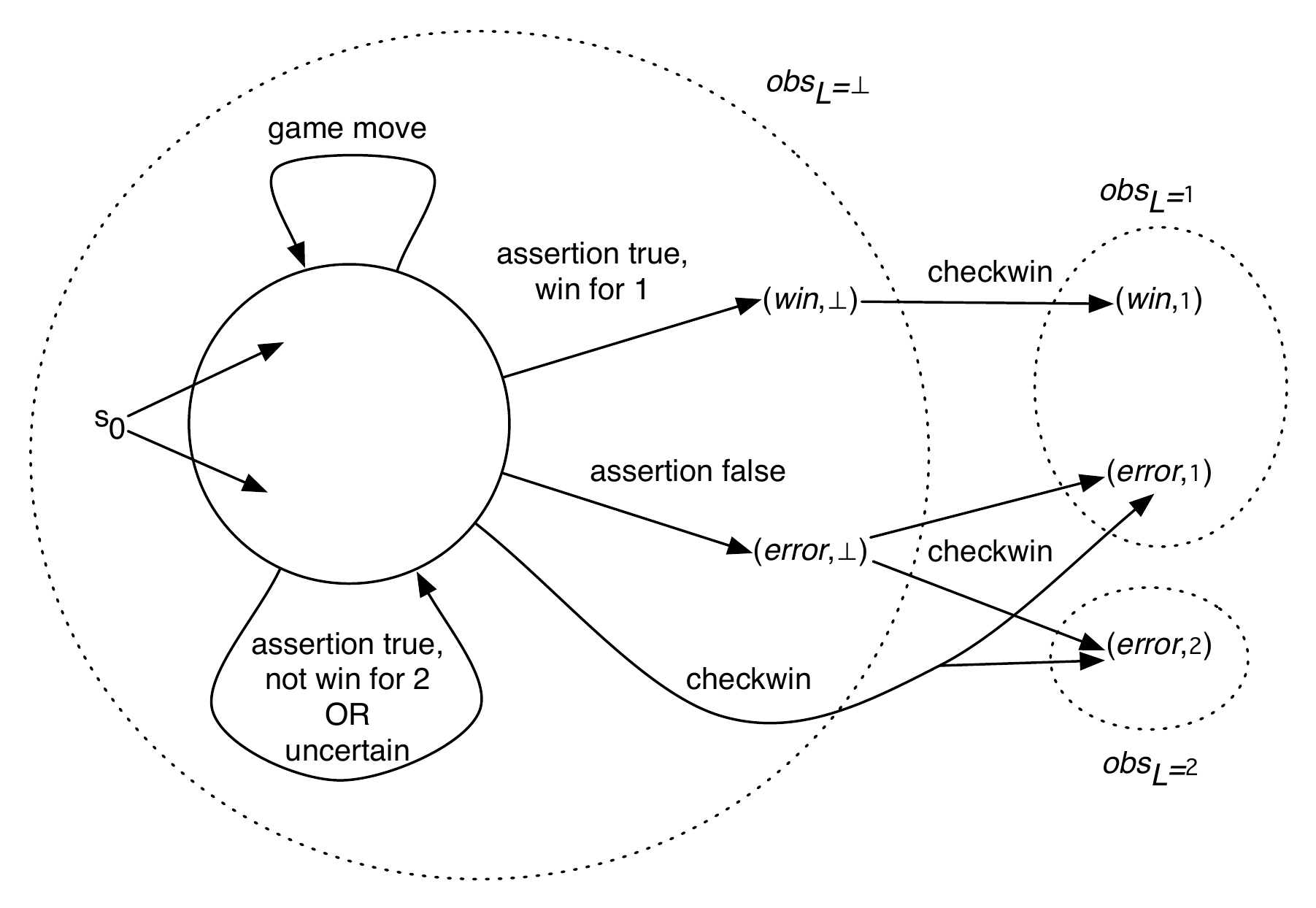}}
\caption{\label{fig:terminal} High level structure of $M(G)$}
\end{figure}

States of $M(G)$ in $\C\times \F$ are used to capture the effect of assertions made
by $H$ relating to the winning conditions, and play a key role in ensuring that
an $L$ view is eliminated under the appropriate conditions.
These states form a ``terminal" part of the machine: it is not possible to return to the
component $\C\times \Plates\times \B\times \lastmove$ from these states.
The $\C$ component simply tracks the simulation stages.
Figure~\ref{fig:terminal} sketches the way that the $\F$ components of these states are
used to check winning conditions for the game and to generate observations.

Intuitively, at various stages of the simulation (viz., $H_0$ and $H_1\ldots H_{h_2}$),
agent $H$ is allowed to make assertions about the state of the game.
When $H$ makes such an assertions, $M(G)$ checks whether they are true at the
plate being simulated.
\begin{itemize}
\item If the assertion entails that there is not
yet a win for player 2, and is true of the present plate, or does not concern the present plate,
then we continue the simulation.
(Specifically, this case occurs when the assertion is that a particular plate is
blocking the hole for player 2 under consideration, and this is true
or concerns another plate.)
\item If the assertion entails a win for player 1  and is true  of the
present plate, then a transition is made to a state $(\win,\bot)$.
(Specifically, the  assertion is that a particular player 1 hole is open, and this holds at the present plate.)

\item The remaining possibility is that the assertion is false.
In this case  we make a transition to a state $(\lose, \bot)$.
(We have this case when either the assertion is that a particular player 1 hole is open, but this is false at the present plate,  or is that the present plate is blocking a particular player 2 hole, but this is false.)
\end{itemize}
When $L$ eventually performs a $\cw$ action in the appropriate phase $L_2$,
states with $(\lose,\bot)$ could produce either the observation $1$ or $2$.
By contrast, states  with $(\win,\bot)$ produce only the observation 1.
Thus, the $\win$ states result in a reduced set of views.
Note that we can only check assertions locally at the current plate, and the winning
condition for player 1 requires that a player 1 hole be unblocked   at all plates.
The way that the encoding handles this is via $L$'s uncertainty about which plate is
being monitored: if there is any plate at which the hole is blocked, then
there will be a run consistent with $L$'s observations, monitoring this plate,
at which we have  $(\lose,\bot)$. When $L$ performs $\cw$ in this run it will obtain
both observations $1$ and $2$,  and the set of views is not reduced.

\subsubsection{Transitions of $M(G)$}

In general, transitions in a machine are labelled by joint actions
$(a_H, a_L)$ of $H$ and $L$, but since $M(G)$ is a scheduled machine,
at most one of these has any effect on the state.  To simplify the
presentation, we use the convention of writing $s\LT{b}t$ with $b\in A_u$, to
indicate that the current transition is dependent only on $u$, and
write $s\LT{\tau}t$ if the transition is independent of both $H$ and
$L$ (i.e., a system step).
To capture the scheduler, we define the function $\nxt: \C \rightarrow
\C$ so that it maps each element of the sequence $L_1, H_0, L_1, \bot,
H_1 \ldots H_{h_2}$ to the element next in the sequence, with
$\nxt(H_{h_2})=L_1$.

\newcommand{\Open}{\mbox{\tt Open}}

To describe the transitions we use the following predicates relating
to the winning conditions. When $i \in \{1,2\}$ is a player, $j=
1..h_i$ is a hole associated to that player, $k \in \Plates$ is a
plate and $b\in \B$ is a plate position, we define $\Open^i_j(k,b)$,
to be true just when either proposition $P_k$ does not occur
positively or negatively in $\gamma^i_j$, or $P_k$ occurs positively
in $\gamma^i_j$ and $b=1$ or $\neg P_k$ occurs in $\gamma^i_j$ and
$b=0$.  Intuitively, this says that when plate $k$ is in position $b$,
it is not blocking hole $j$.

We group the transitions according to the step of the schedule.
We first describe transitions from states in $\{s_0\} \cup
(\C\times\Plates\times\B\times\lastmove)$.

\paragraph{\textbf{Initial Step $\bot$}}
The initial state of $M(G)$ is $s_0$.  From this state,
we non-deterministically choose one plate to be monitored.  
This transition does not depend on any agents and has the form
$$s_0\xrightarrow{\ \tau \ } (L_1,i, \nu_0(i), \bot)$$ 
for $i \in [1\dots n]$, where $\nu_0$ is the initial state of $G$.

\paragraph{\textbf{Stage $L_1$}}
At this stage, we simulate player~1's move of a plate $j$. The special
action $\cw$ is used for the pass move.  For each $\mmove_j \in A_L$,
and state $(L_1,i,k,a)$ there is a transition of the form:
$$(L_1,i,k,a) \xrightarrow{\ \mmove_j \ } (H_0,i,k',\mmoved_j)$$ with
$k'=1-k$ if $j=i$ and $k'=k$ otherwise.  We also add $$(L_1,i,k,a)
\xrightarrow{\ \cw \ } (H_0,i,k,\pass)$$ for the pass move.

\paragraph{\textbf{Stage $H_0$}}
At this stage $H$ may try to prove that player~1 can peek through some hole
$j \in [1..h_1]$.  To do this, it chooses an action $\isopen_j$.  As
it is a guess, $H$ might be wrong.
If $H$ claims that player~1 can peek through hole $j$ and this is right
at the present plate, we reach a ``winning'' state in $\C\times\F$, otherwise an error state.
From a plate simulation state $(H_0,i,k,a)$ there is a transition
\[
(H_0,i,k,a) \xrightarrow{\ \isopen_j \ } (L_2,r,\bot)
\]
where $r=\win$ if $\Open^1_j(i,k)$,
and  $r=\lose$ otherwise. $H$ may also intentionally
choose not to declare a win, by performing any of its actions
$\isblocked_j$ for $j\in\set{1\dots n}$.
This is captured by the transitions
\[(H_0,i,k,a)\xrightarrow{ \ \isblocked_j \ } (L_2,i,k,\bot)\]
for $j\in\set{1\dots n}$.

\paragraph{\textbf{Stage $L_2$}}
At this stage, $L$ can perform the action ``$\cw$'' to check if
$H$ has proved a win by player 1. If the current state at this stage is a plate simulation
state, then $H$ has not yet claimed a win for player 1,
and  any past assertions made by $H$ about player 2's winning condition were either true or irrelevant to the
current plate. In this case, we do not have evidence for a player 1 win,
so we do not  wish to eliminate an $L$ view. Thus, from states $s=(L_2,i,k,a)$,
we have transitions
 \[s \xrightarrow{\ \cw \ } (\bot,\lose,1) \quad \text{and} \quad s \xrightarrow{\ \cw \ } (\bot,\lose,2)\]
 so that both observations 1 and 2 can be obtained.
Agent $L$ is also allowed to continue playing without checking
for a win, by performing any of the actions  $\mmove_j$ with $j\in[1\dots n_1]$.
For this case we have a transition
\[(L_2,i,k,a) \xrightarrow{\ \mmove_j\ } (\bot,i,k,\bot)~\mathpunct. \]
(We discuss the case of  transitions at this stage from states in $\C\times \F$ below.)

\paragraph{\textbf{Stage $\bot$}}
At this stage, we simulate a move of a plate by player~2.  The
following transitions are in $M(G)$:
\[
(\bot,i,k,a) \xrightarrow{\ \tau\ } (H_1,i,k',\mmoved_j)
\]
for each $i \in \Plates$, $k\in \B$
and  $j \in [(n_1+1)\dots n]$
(player~2's plates), and: $k'=1-k$ if $i=j$ and $k'=k$ otherwise.
To model a pass move by player 2 we also have a transition
\[
(\bot,i,k,a) \xrightarrow{\ \tau\ } (H_1,i,k,\pass)~\mathpunct.
\]

\paragraph{\textbf{Stages $H_1$ to $H_{h_2}$}}
In these stages, $H$ tries to prove that
that last move by player 2 was not a winning move
for player 2. It does so by showing that all player~2 holes
are blocked (by at least one plate).
For each player 2 peek hole $j$, at stage $H_j$, agent
$H$ chooses a plate $i\in[1\ldots n]$ and asserts that the hole is blocked by that plate
using the action $\isblocked_i$. An incorrect assertion results in a transition
to an error state $(c, \lose, \bot)$. In particular, if the current state is
indeed a win of player~2, then some hole $j$ is open at all plates, and
any attempt $H$ makes to assert that it closed at a plate causes a transition to an error state.

This is encoded by the following transitions.
At state $(H_j,i',k,a)$ with $j \in [1..h_2]$,
we have
\[
 (H_j,i',k,a) \xrightarrow{\ \isblocked_i \ }
 (\nxt(H_j),i',k,\bot) ~\mathpunct.
\]
 when either $i\neq i'$ (plate $i$ is not monitored in the current state),
  or $i=i'$ and not $\Open^2_j(i,k)$ (the present plate is blocking player 2's hole $j$).
On the other hand, if $i = i'$ and $\Open^2_j(i,k)$, i.e.,  
player~2's peek hole $j$ is not blocked by plate $i$, we have
the transition
\[
 (H_j, i', k, a) \xrightarrow{\ \isblocked_i \ } (\nxt(H_j),\lose,\bot) ~\mathpunct.
\]

As pointed out before, the construction is designed to ensure that $H$
knows the exact state of the game and thus can always determine
whether a peek hole is blocked or not.
Since we are looking for a winning strategy for player 1, if there is
a win by player 2 then the present $L$ strategy has failed.  We would
therefore like to insist that $H$ must play only $\isblocked_i$
actions at this stage of the process.  To ensure this, we define
transitions so that all $H$ actions other than $\isblocked_i$ cause an
error transition at these stages. That is, for all $H$ actions 
$\isopen_j$ and states of the form $(c, i,k, a)$ with $c\in \{H_1,\ldots , H_{h_2}\}$,
we have a transition
\[
(c,i,k,a) \xrightarrow{\isopen_j}
(\nxt(c), \lose,\bot) ~\mathpunct.\]

This completes the description of transitions from states in $\C\times\Plates\times\B\times\lastmove$.

\paragraph{\textbf{Transitions from $\C\times \F$}}
The behaviour of the machine on states in $\C\times \F$
was described informally above. The main effect of actions
is from $\cw$ actions performed at stage $L_2$.
From states with observation $\bot$, the action $\cw$ causes an $L$ observation of $1$ or $2$.
For the $\win$ states we have a transition
\[(L_2, \win, \bot)\LT{\cw}(\bot, \win, 1)\]
 and for the $\lose$ states we have transitions
\[(L_2, \lose, \bot)\LT{\cw}(\bot, \lose, 1)\quad\mbox{and}\quad (L_2,
\lose, \bot)\LT{\cw}(\bot, \lose, 2)~\mathpunct.\]
For all other cases, i.e., for $c= L_2$ and an action $b = \mmove_j$,
or for $c\in \C\setminus \{L_2\}$ and an action $b$
(appropriate to stage $c$),   we have for all $r\in\set{\win, \lose}$
a transition
\[(c, r, \bot)\LT{b}(\nxt(c), r, \bot)~\mathpunct.\]
States at which an observation of $1$ or $2$ has already been obtained by $L$
act as sinks, except for scheduler moves, i.e., we have a transition
\[(c, r, x)\LT{b}(\nxt(c), r, x)\]
for all $c\in \C$, $r\in \set{\win, \lose}$ and $x\in \set{1,2}$.

\subsubsection{Correctness of the construction}

We now give the argument for the correctness of the encoding.

We first characterize the views obtained by the agents in $M(G)$.
In the case of agent $L$, the structure of the possible views follows straightforwardly
from the fact that the transitions as defined above follow the
structure indicated in Figure~\ref{fig:terminal}. Until a $\cw$ action is performed by $L$
at some stage $L_2$ state, $L$ observes $\bot$. Once it performs that action at
this stage, it will observe either $1$ or $2$ for the remainder of time.
Thus, $L$ views are prefixes of the sequences generated by the regular expressions
\[\bot\, A_L\, \bot\, (A_L\, \bot\, A_L\, \bot\,A_L^- \bot A_L \bot (A_L\,\bot)^{h_2})^*\, A_L\, \bot\, A_L\, \bot \,\cw \,  x\, (A_L \, x)^* \]
with $x=1$ or $x=2$. Here the expression $A_L\, \bot\, A_L\, \bot\,A_L^- \bot A_L \bot (A_L\,\bot)^{h_2}$
corresponds to a round in which $L$ does not
perform $\cw$ at stage $L_2$. If $\alpha$ is an $L$ view, we write
$\Lambda(\alpha)$ for the subsequence of player 1 actions performed at times
when the simulation is at stage $L_1$, where we treat a $\cw$ at such a time
as the action $\pass$.

The $H$ views are prefixes of the sequences in the regular expression
\[ \bot (A_H \, (\Move_1\cup \Move_2\cup \set{\bot}))^* (A_H\,\doneobs)^*\]
where the observations obtained at stage $H_0$ are in $\Move_1\cup \set{\bot}$,
the observations obtained at stage $H_1$ are in $\Move_2\cup \set{\bot}$, and
all other observations before the first $\doneobs$ are $\bot$.

We will show that there is a correspondence
between plays of the game $G$ and views $\beta$ of $H$.
In particular, given an $H$ view $\beta$, let $\sigma(\beta) = \lambda_1 \ldots \lambda_n$
be the subsequence of elements of $\Move_1\cup \Move_2$
appearing in observations. It follows from the
definition of the transition relation that $\sigma(\beta)$ is
an alternating sequence of player 1 and player 2 actions.
Since the moves of $G$ have a deterministic effect on the states of $G$,
we obtain a play
$$\varrho(\beta) = \nu_0 \LT{\lambda_1} \nu_1 \LT{\lambda_2} \nu_2 \ldots \LT{\lambda_n} \nu_n$$
of the game $G$. We define $\nu(\beta)$ to be the final state $\nu_n$ of $\varrho(\beta)$.

Consider an action $a$ of $H$ performed at an $H$ view $\beta$ at stage $c$. This will
be recorded in the view of $H$, which will have the form $\beta\, a \, o$
immediately after this action. We say that the action $a$ is {\em truthful} at
view $\beta$, if
\begin{itemize}
\item $\beta= \bot$ is the view obtained at the run $s_0$, or
\item $\beta$ contains an action $\isopen_j$ (intuitively, $H$ has already discharged the obligation to
prove that $L$ wins), or
\item
the stage $c$ is in $\set{L_0, L_1, \bot}$ ($H$ makes no assertion at these stages), or
\item $c= H_0$ and $a = \isopen_j$ and hole $j$ of player 1 is open in state $\nu(\beta)$,
i.e., $\nu(\beta)$ satisfies $\gamma^1_j$, or
\item $c= H_0$ and $a = \isblocked_j$ for some $j$ (this corresponds to
no assertion by $H$), or
\item $c = H_k$ for $k\in [1\ldots h_2]$ and $a = \isblocked_i$ and not
$\Open^2_k(i, \nu(\beta)(i))$, i.e., player 2's hole $k$ is blocked at plate
$i$ in state $\nu(\beta)$.
\end{itemize}
Note that we omit the case where $c=H_k$, $k\geq 1 $ and $a = \isopen_j$;
intuitively, this corresponds to an assertion of False by $H$, which has
an obligation to prove that the play is winning for player 1, and is failing to
do so in this instance.
We say that the view $\beta$ is truthful if for every prefix
$\beta\, a \, o$, the action $a$  is truthful at $\beta$.

We now show that, so long as the  play $\varrho(\beta)$ is undecided,
the knowledge set of $H$ encodes the state $\nu(\beta)$.
For agent $u$ we write $K_u(\alpha)$ for the set of final states of runs $r$ of $M(G)$
with $\view_u(r) = \alpha$, representing agent $u$'s knowledge of the
state after obtaining view $\alpha$.
Given a stage $c$, a state $\nu$ of game $G$ and $a \in \lastmove $, we let $S(c,\nu,a)$ be the set defined
by  $S(c,\nu,a)=\set{(c,i,\nu(i),a)~|~i \in [1..n]}$. Note that $\nu$ can be recovered from the
set $S(c,\nu,a)$.

\begin{proposition} \label{prop:undecided}
Suppose that $\beta$ is an $H$ view in $M(G)$ of length at least 1 such that
$\beta$ does not contain the observation $\doneobs$ and
$\varrho(\beta)$ is an
undecided play of game $G$. Let $c$ be the stage reached at the end of $\beta$ and
let the final observation of $\beta$ be $a$. Then if $\beta$ is truthful,  we have
$K_H(\beta) = S(c, \nu(\beta),a)$.
\end{proposition}

\begin{proof}
We proceed by induction on the length of $\beta$.
If $\beta$ has length 1, then it arises from a run $$s_0 \LT{(a_H,a_L)} (L_1,i, \nu_0(i), \bot)$$
so we have $\beta = \bot a_H\bot$, and $\sigma(\beta)$ is the empty sequence, corresponding to the play
$\varrho(\beta) = \nu_0$.
(Since we are interested in views, we
use the explicit form of runs here, with actions of both $H$ and $L$ given,
rather than the shorthand form used above, which mentioned only actions of the scheduled agents
and implies that action the other agents may take all possible values.)
The runs consistent
with $\beta$ are
$$s_0 \LT{(a_H,a)} (L_1,j, \nu_0(j), \bot)$$
where $j\in \Plates$ and $a$ is some action of $L$. It is immediate that the claim holds.

Inductively, assume that $\beta \, a_H \, o$ is a truthful $H$ view in $M(G)$ such that
$K_H(\beta) = S(c, \nu(\beta),a)$ for some $c\in \C$ and $a\in \lastmove$.
We need to show that $K_H(\beta\,a_H\, o) = S(c', \nu(\beta\, a_H\, o),a')$ for some $c'\in \C$ and $a'\in \lastmove$.
Note that
$$K_H(\beta\,a_H \, o) = \set{t~|~s\in K_H(\beta)\text{ and } a\in A_L \text{ and } s\LT{(a_H,a)} t\text{ and }\obs_H(t) = o} ~\mathpunct.$$
We consider each of the possible cases of the stage $c$.

\paragraph{Stage $c=L_1$}
Here we have either $o = \mmoved_j$ or $o= \pass$, which
records the action of $L$, and the transition does not depend on $a_H$.
Note that $\sigma(\beta \, a_H \, o) = \sigma(\beta)\, o$ in this case.
Let  $\nu(\beta) \LT{o} \nu'$, so that $\nu(\beta\, a \, o) = \nu'$.
For each plate $i\in \Plates$, we have $(L_1,i,\nu(i),a) \in  K_H(\beta)$,
and there is a unique transition
$$(L_1,i,\nu(i),a) \LT{(a_H,o)} (H_0,i, \nu'(i), o)$$
yielding observation $o$ at the next state.
It follows that $K_H(\beta\,a_H\, o) = S(H_0, \nu(\beta\, a_H\, o),o)$.

\paragraph{Stage  $c=H_0$} Here we have $o = \bot$,
and $\sigma(\beta\,a_H \, o) = \sigma(\beta)$ and $\nu(\beta\, a\, o) = \nu(\beta)$.
Since  $\varrho(\beta \, a \, o)$ is
undecided, $\nu(\beta)$ is not a winning position for player 1.
Because $\beta\,a_H \, o$ is truthful, we cannot have
that $a_H$ is $\isopen_j$ for any $j$, so we must have that
$a_H = \isblocked_j$ for some $j$.
For each plate $i\in \Plates$, we have $(H_0,i,\nu(\beta)(i),a) \in  K_H(\beta)$,
and there is a transition
$$(H_0,i,\nu(\beta)(i),a) \LT{(a_H,a')} (L_1,i, \nu(\beta)(i), \bot)$$
yielding observation $\bot$ at the next state for each action $a'$ of $L$.
It follows that $K_H(\beta\,a_H\, o) = S(H_0, \nu(\beta\, a_H\, o),o)$.

\paragraph{Stage  $c=L_2$} Since $\beta$ does not contain $\doneobs$, here we have $o = \bot$,
and $L$ cannot have performed the action $\cw$.
Also $\sigma(\beta\,a_H \, o) = \sigma(\beta)$ and $\nu(\beta\, a\, o) = \nu(\beta)$.
Thus,  from each  $(L_2,i,\nu(\beta)(i),a) \in  K_H(\beta)$,
we have  a transition
$$(L_2,i,\nu(\beta)(i),a) \LT{(a_H,a')} (\bot,i, \nu(\beta)(i), \bot)$$
from which we obtain that $K_H(\beta\,a_H\, o) = S(\bot, \nu(\beta\, a_H\, o),\bot)$.

\paragraph{Stage  $c=\bot$} Here we have that $o = \pass$ or $o = \mmoved_i$
for some plate $i = n_1+1\ldots n$ of player 2.
Thus, $\sigma(\beta\,a_H \, o) = \sigma(\beta)\, o$ and
if $\nu(\beta) \LT{o} \nu'$ then $\nu(\beta\, a\, o) = \nu'$.
For each plate $i\in \Plates$, we have $(\bot,i,\nu(\beta)(i),a) \in  K_H(\beta)$,
and there is a unique transition
$$(\bot,i,\nu(\beta)(i),a) \LT{(a_H,o)} (H_1,i, \nu'(i), o)$$
yielding observation $o$.
It follows that $K_H(\beta\,a_H\, o) = S(H_1, \nu(\beta\, a_H\, o),o)$.

\paragraph{Stage  $c=H_i$ for $i= 1\ldots h_2$}
Here we have $o = \bot$,
and $\sigma(\beta\,a_H \, o) = \sigma(\beta)$ and $\nu(\beta\, a\, o) = \nu(\beta)$.
Since $\beta\,a_H \, o$ is truthful and $\sigma(\beta\,a_H \, o)$ is undecided,
the position is not winning for player 2. Thus, we must have that $a_H = \isblocked_j$
for some $j$ such that not $\Open^2_i(j,\nu(\beta)(j))$.
For each plate $i'\in \Plates$, we have $(H_i,i',\nu(\beta)(i'),a) \in  K_H(\beta)$,
and there is a transition
$$(H_i,i',\nu(\beta)(i),a) \LT{(a_H,a)} (\nxt(H_i),i', \nu(\beta)(i'), \bot)$$
yielding observation $\bot$ for every $L$ action $a$.
It follows that $K_H(\beta\,a_H\, o) = S(H_1, \nu(\beta\, a_H\, o),o)$. \qed
\end{proof}

In fact, we can show this characterization of $K_H(\beta)$ in one
further case, corresponding to the state of the simulation just after
player 1 plays a winning move.

\begin{proposition} \label{prop:undecided-final}
Suppose that $\beta$ is an $H$ view in $M(G)$ of length at least 1 such that
$\beta$ does not contain the observation $\doneobs$ and
$\varrho(\beta)$ is a play in which the last move is a move of player 1 by which
player 1 wins the game. Assume that no shorter prefix of $\beta$ has this property and
let the final observation of $\beta$ be $a$. Then if $\beta$ is truthful,  we have
$K_H(\beta) = S(H_0, \nu(\beta),a)$.
\end{proposition}

\begin{proof}
The minimality constraint on $\beta$ implies that $\beta = \beta' a_H o$
where $\beta'$ is at stage $L_1$,  that $\varrho(\beta')$ is not a
winning play for either player, and that $o\in \Move_1$ is a move of player~1
such that $\varrho(\beta' a_H o) = \varrho(\beta') \LT{o} \nu(\beta)$.
By Proposition~\ref{prop:undecided}, we obtain that
$K_H(\beta') = S(L_1, \nu(\beta),a)$ for some $a$. The argument for the case of
$c= L_1$ in the proof of Proposition~\ref{prop:undecided-final} now yields
the conclusion. \qed
\end{proof}

\begin{proposition} \label{prop:LH}
Suppose that $r$ is a run of $M(G)$ and $s\in K_H(\view_H(r))$ is a state in $\C\times\Plates\times\B\times\lastmove$.
Then there exists a run $r'$ of $M(G)$ with final state $s$ such that $\view_H(r') = \view_H(r)$ and $\view_L(r') = \view_L(r)$.
\end{proposition}

\begin{proof}
By induction on the length of $r$. For $r=s_0$, the
statement is trivial, since we must have $K_H(\view_H(r)) = \set{s_0}$.
For another base case, suppose $r$ is the initial step of a run.
In this case, $$r= s_0 \LT{(a_H,a_L)} (L_1,i,\nu_0(i),\bot)$$ and
$\view_H(r) = \bot \, a_H\,  \bot$ and $\view_L(r) = \bot \, a_L\,  \bot$.
If $s\in K_H(\view_H(r))$, then we must have
$s= (L_1, j,\nu_0(j), \bot)$ for some $j$. We may
take $$r' =  s_0 \LT{(a_H,a_L)} (L_1,j,\nu_0(j),\bot)$$ and this has the
required properties.

For the induction, let $$r = r_1 \LT{(a_H,a_L)} t$$ where the result holds for $r_1$,
which is at stage $c$.
Let $s\in K_H(\view_H(r))$.
Since $s\in \C\times\Plates\times\B\times\lastmove$, the final observation of
$\view_H(r)$ is not $\doneobs$, and it follows that $t\in \C\times\Plates\times\B\times\lastmove$ also.
We have that $s = (\nxt(c), i, k,a)$ for some $i,k,a$, and arises
in  $K_H(\view_H(r))$ from some run $$r_2\LT{(a_H', a_L')} s$$ with
$\view_H(r_2) = \view_H(r_1)$ and $a_H' = a_H$ and $\obs_H(s) = \obs_H(t)$.
Necessarily, the final state of $r_2$ is in $\C\times\Plates\times\B\times\lastmove$.
Moreover, it is in $K_H(\view_H(r_2)) = K_H(\view_H(r_1))$.
By the induction hypothesis, there exists a run $r_3$ ending in the same final
state as $r_2$, with $\view_H(r_3) = \view_H(r_1)$ and $\view_L(r_3) = \view_L(r_1)$.
We consider the possibilities for the
scheduler step $c$ of $r_1$:
\begin{itemize}
\item {\em Case $c=L_1$}:
Note that at this stage, the final $L$ action in $r$ can be deduced from
$\obs_H(s) = \obs_H(t)$, so in fact we have $a_L' = a_L$ also.  Let $$r' = r_3
\LT{(a_H,a_L)} s~\mathpunct.$$ This is a run because $r_3$ and $r_2$ have the
same final state, and the final transition in $r'$ is identical to the final
transition of the run $r_2$. Then  $\view_H(r') = \view_H(r_3)\, a_H \,
\obs_H(s) = \view_H(r_1) \, a_H \, \obs_H(s) = \view_H(r)$ and $\view_L(r') =
\view_L(r_3)\,  a_L\,  \bot  = \view_L(r_1)\,  a_L \, \bot = \view_L(r)$, as
required.

\item {\em Case $c = H_k$ for $k = 0..h_2$}:  Transitions at these stages are independent of
$L$, so we can switch the action of $L$ in any transition label while keeping the states the
same. So $$r' = r_3 \LT{(a_H,a_L)} s$$
is a run and satisfies the required properties.

\item {\em Case $c = L_2$}:
Here it follows from $t\in \C\times\Plates\times\B\times\lastmove$ that
$a_L \neq \cw$. Hence $a_L = \mmove_j$ for some $j$, and $\obs_H(t) =  \obs_H(s) =\bot$.
For the same reasons, $a_L' = \mmove_{j'}$ for some $j'$.
The transitions for $\mmove_j$ and $\mmove_{j'}$ at this stage are identical.
We may therefore take
$$r' = r_3 \LT{(a_H,a_L)} s$$
and this is a run and satisfies the required properties.

\item{\em Case $c= \bot$}:  Here transitions are independent of both players,
so $$r' = r_3 \LT{(a_H,a_L)} s$$
is a run and satisfies the required properties.
\end{itemize}
This completes the proof of the inductive case. \qed
\end{proof}

We can now prove the key result that shows that $G$ has a winning strategy for player~1 iff
$M(G)$ satisfies $\tNDS$.

\begin{lemma}
  There exists  a
  winning strategy for player~1 in $G$
  iff there is an $L$  view $\alpha$ of $M(G)$ and an $H$ strategy $\pi$ that excludes
  $\alpha$ in $M(G)$.
\end{lemma}

\newcommand{\skp}{a_0}

\begin{proof}\mbox{}
  \paragraph{Only If Part}
  Assume player~1 has a winning strategy  in $G$.
  As argued before, this strategy can be given by the list $\Lambda =\lambda_1, \lambda_3 \ldots $
  of moves of player~1.
  The number of moves player~1 needs to win is bounded: indeed, in
  every play of the game, a winning position for player 1 is
  eventually reached, and no winning position for player 2 is reached
  before this position. By Koenig's lemma, there exists a number $N$
  such  in all plays of the game compatible with $\Lambda$, player 1 has won the game
  at the latest, just after the $N$-th move. Thus, we may assume that
  $\Lambda =  \lambda_1, \lambda_3 \ldots \lambda_N$.

  We can prove that there exists an $H$ strategy $\pi$ in $M(G)$ that
  excludes the $L$ view
\[
\alpha = \bot \,\skp\, \bot\,  a_{1}\, \bot\,  (\skp\, \bot)^{3+h_2}\, a_{3}\, \bot \,
(\skp\, \bot)^{3+h_2} \ldots a_{N}\, \bot\, \skp\, \bot\, \cw\, 2.
\]
where $\skp$ denotes any letter in $A_L^-$,
and each $a_i$ for $i$ odd is the $L$ action that corresponds to $\lambda_i$
at stage $L_1$, i.e., $a_i= \lambda_i$ if $\lambda_i = \mmove_j$ for some $j$, and
$a_i= \cw$ if $\lambda_i = \pass$. The strategy $H$ is defined as follows.
For $H$ views $\beta$ such that the sequence
of player 1 moves in $\sigma(\beta)$ is a prefix of $\Lambda$,
we let $\pi(\beta)$ be any truthful action of $H$
at $\beta$. In all other cases, $\pi(\beta)$ is chosen arbitrarily.

We first need to show that $\pi$ is well-defined. For this, we need to
show that $H$ is able to act truthfully whenever
the sequence  of player 1 moves in
 $\sigma(\beta)$ is a prefix of $\Lambda$.
 Suppose, therefore, that the sequence  of player 1 moves in
 $\sigma(\beta)$ is a prefix of $\Lambda$. Then the play $\varrho(\beta)$ is
 not winning for player 2, since $\Lambda$ is a winning strategy for
 player 1.  There are two possibilities: the play is undecided, or the
 play is winning for player 1. If the play $\varrho(\beta)$ is undecided, then
 by Proposition~\ref{prop:undecided},   we have that
 $K_H(\beta) = S(c, \nu(\beta),a)$. Since $\nu(\beta)$ is, by definition,  the final state of
 $G$ reached in the play $\sigma(\beta)$, if $\sigma(\beta)$
 ends in a move of player 1 then  $\nu(\beta)$ is not a winning state for player 1, and
 if $\sigma(\beta)$  ends in a move of player 2 then  $\nu(\beta)$ is not a winning
 state for player 2. In either case, depending on the stage, it is possible to select an
 action that is truthful at $\beta$.

 In the other case, the play $\varrho(\beta)$ is already winning for player 1.
 Let $\beta'$ be the smallest prefix of $\beta$ such that $\sigma(\beta')$ is
 winning for player 1. Since the last action in $\sigma(\beta')$ is a move of
 player 1 (we can assume without loss  of generality that the game is
 undecided at the initial state),  must have that $\beta'$ is at stage $H_0$,
 and $\nu(\beta)$ is a winning state for player 1.
In the case that $\beta' = \beta$, we choose $\pi(\beta)$
to be any action $\isopen_j$ such that $\nu(\beta)$ satisfies $\gamma^1_j$,
i.e., player 1's hole $j$ is open at all plates in $\nu(\beta)$. This is then a truthful
action at $\beta$.  In all other cases,
we choose $\pi(\beta)$ arbitrarily. (Note that, by definition, after $H$'s first $\isopen_j$, any choice
of $H$ action is truthful.)

We now argue that $\pi$ excludes view $\alpha$. To the contrary, suppose that
$r$ is run consistent with $\pi$ and $\view_L(r) = \alpha$. Consider $\beta = \view_H(r)$, and write this
as $\beta = \beta_1\, a_H\, o$. Then the sequence  of player 1 moves in $\sigma(\beta_1)$
is $\Lambda$. Since $\Lambda$ is a winning strategy for player 1, the play $\varrho(\beta_1)$ is a winning
play for player 1. Consider the shortest prefix $\beta_2$ of $\beta_1$ such that
$\sigma (\beta_2)$ is a winning play for player 1. Then $\beta_2$ is at stage $H_1$, and
there is at least one $H$ action  $a$ and observation $o'$ such that $\beta_2 \, a\, o'$ is a prefix of $\beta$.
(In the worst case, $\beta_2 = \beta_1$ and $a= a_H$.) By construction of $\pi$,
$a$ is an action $\isopen_j$ that is truthful at $\beta_2$. Using Proposition~\ref{prop:undecided-final},
$K_H(\beta_2) = S(H_1, \nu(\beta_2), a')$ for some $a'$. Because
$\isopen_j$ is truthful at $\beta_2$, we obtain that  $K_H(\beta_2 \, a\, o') = \set{(\bot,\win,\bot)}$.

In particular, the prefix $r_2$ of $r$ with
$\view_H(r_2) = \beta_2\, a\, o' $ has final state $(\bot,\win,\bot)$.
Since $L$ does not perform  $\cw$ at stage $L_2$ in the interim,
the prefix $r_1$ of $r$ with $\alpha = \view_L(r_1)\, \cw \, 2 $
has final state $(L_2,\win,\bot)$. But then
we get that the final state of $r$, after $L$ performs $\cw$,
is  $(\bot,\win,1)$, which yields an $L$ observation  of 1 rather than the
final observation 2 of $\alpha$. This is a contradiction.

\paragraph{If Part}
  Assume there is no  winning strategy for player~1 in $G$.
  We show that there is no $H$ strategy that can exclude any $L$ view.
  To the contrary, assume an $L$ view $\alpha$ that is excluded by an $H$ strategy
  $\pi$. The following must hold:
  \begin{enumerate}
  \item The view $\alpha$ contains a $\cw$ action at stage $L_2$.
    Indeed, if no $\cw$ action at stage $L_2$ occurs,
    then all $L$ observations in the view must be $\bot$.
    For any such sequence of $L$ actions,
    there is always a run consistent with $\pi$
    yielding $L$ observation $\bot$ throughout, so that, contrary to
    assumption, $\alpha$ is not excluded by $\pi$.

  \item The view $\alpha$ is not of the form $\alpha_1 \, \cw \; 1 (A_L \,1)^*$, with
  $\alpha_1$ being a view at stage $L_2$, and containing no prior $\cw$ action at stage $L_2$.
  There is always a run consistent with $\pi$ that yields such a view.
  There are two possibilities. If the final state $s$ of a run yielding $L$ view $\alpha_1$
  is in $\C\times\Plates\times\B\times\lastmove$, then the $\cw$ action extends this
  run to one yielding $L$ view  $\alpha_1 \, \cw \; 1$ by means of the stage $L_2$ transition
  $$s \LT{\cw} (\bot, \win, 1)~\mathpunct.$$ Otherwise, the state $s$ is in $\C\times \F$,
  and must be of the form $(L_2,r,\bot)$, for $r\in \set{\win,\lose}$, since there has not yet been a  $\cw$ at stage $L_2$.
  In this case, we obtain observation $1$ at the next step by means of a transition
  $$(L_2,r,\bot) \LT{\cw} (\bot, r,1)$$ for both possible values of $r$.
  The resulting states with observation $1$ are sinks, so we can extend these runs to obtain a
  run with $L$ view $\alpha_1 \, \cw \; 1 (A_L \,1)^*$.
  (In either case, we may take the $H$ action in the final transition to be the action prescribed by $\pi$, since
  this transition is independent of $H$.)
  \end{enumerate}
  It follows that view $\alpha$ is in the regular set $\alpha_1 \, \cw \; 2 (A_L \,2)^*$, with
  $\alpha_1$ being a view at stage $L_2$, and containing no prior $\cw$ action at stage $L_2$.

  Let $\Lambda = \lambda_1 \lambda_3 \cdots \lambda_N $ be the player~1 moves of $G$
  corresponding to the actions taken by $L$ in $\alpha$ at each  $L_1$ stage.
  (These are the same as the $L$ actions, except that we treat $\cw$ at stage $L_1$ as corresponding to $\pass$.)
  The sequence $\Lambda$ is a player~1  strategy  in $G$.
  This strategy cannot be winning for player~1  as  we have assumed that
  this player has no winning strategy in $G$.
  Thus, there exists some sequence $\lambda_2 \ldots \lambda_{N-1}$ of player~2 moves such that the
  play $$\varrho = \nu_0 \LT{\lambda_1} \nu_1 \LT{\lambda_2} \nu_2 \LT{\lambda_3} \ldots  \LT{\lambda_N} \nu_N$$
  is not winning for player~1. Let $r_1$ be a run consistent with $\pi$ such that $\view_L(r_1) = \alpha_1$
  and at the $m$-th occurrence of stage $\bot$, we take the transition
  $$(\bot, i, k,a) \LT{\tau} (H_1, i, k',\lambda_{2m}) $$
  corresponding to move $\lambda_{2m}$ by player~2.
  (The choices of $L$ actions in this run come from $\alpha_1$, and the
  choices of $H$ actions are fixed by the strategy $\pi$. The
  only nondeterminism remaining is in the initial step,
  where  we choose the plate $i$ to be monitored in the simulation. Since we will work at the
  level of the $H$ view, any choice suffices.)
  Let $\beta_1 = \view_H(r_1)$ be the $H$ view obtained along this run.
  Note that by construction of $r_1$, we obtain that
  $\varrho(\beta_1) = \varrho$ is the play which is not winning for player 1.

  We now argue that the  view $\beta_1$ is truthful and the play $\varrho(\beta_1)$ is also not winning for player 2.
  More precisely, we claim that
  for every prefix $\beta' a_H o$ of $\beta_1$, we have (1) the action $a_H = \pi(\beta')$ is truthful at $\beta'$
  and (2) the play $\varrho(\beta' a_H o)$ is not winning for player 2.
  We proceed by induction, assuming that $\beta'$ is truthful and  $\varrho(\beta')$ is not winning for player 2.
  Note that since $\varrho(\beta')$ is also not winning for player~1,
  we obtain by Proposition~\ref{prop:undecided} that $K_H(\beta') = S(c,\nu(\beta'), a)$ for some $c\in \C$ and $a\in \lastmove$.
  We consider the possible cases for the stage $c$:
  \begin{enumerate}
  \item If $c = L_1$, then  $\varrho(\beta' a_H o)= \varrho(\beta') \LT{\lambda_m} \nu_m$
  for some player~1 move $\lambda_m$.
  Since $\varrho(\beta')$ is not winning for player~2, an extension by a player~1 move also cannot
  be winning for player~2. But also $a_H$ is trivially truthful at $\beta'$, so both (1) and (2) hold.

  \item If $c=H_0$, then  $\varrho(\beta' a_H o)= \varrho(\beta')$ is not
  winning for player 2 by assumption, so we have (2). For (1),
  note that if $a_H = \isblocked_j$ for some $j$ then $a_H$ is trivially truthful at $\beta'$.
  We show that the other case, where $a_H = \isopen_j$ for some $j$, is not possible, because it
  leads to a contradiction.  Note $\varrho(\beta')$  is also not
  a winning play for player~1, so $\nu(\beta')$ is not a winning position for player 1,
  and there exists a plate $i\in \Plates$ for which not $\Open^1_j(i,\nu(\beta')(i))$.
  This means that for the state $s= (c,i,\nu(\beta')(i),a) \in S(c,\nu(\beta'), a) = K_H(\beta')$, we
  have a transition $$s \LT{\isopen_j} (L_1, \lose, \bot)~\mathpunct.$$ It follows
  using Proposition~\ref{prop:LH}  that there exists a run
  $r'$ ending in state $(L_1, \lose, \bot)$ with $\view_H(r') = \beta' a_H o$ and $\view_L(r')$ a prefix of $\alpha_1$.
  The run $r'$ is necessarily consistent with $\pi$ because $\beta$ is consistent with $\pi$.
  Following the actions dictated  for $L$ and $H$ by $\alpha_1$ and $\pi$, respectively,
  we may extend this to a longer run $r'_1$, still consistent with $\pi$,
  with $\view_L(r'_1) = \alpha_1$, also ending
  in state  $(L_1, \lose, \bot)$. But then the next $\cw$ step allows a transition to
   $(\bot, \lose, 2)$, and we obtain a run with $L$ view $\alpha$, a contradiction.

  \item If $c = L_2$, then $a_H$ is trivially truthful, and $\varrho(\beta' a_H o)= \varrho(\beta')$ is not
  winning for player 2 by assumption.

  \item For stages $c = H_k$ with $k \in \set{1\ldots h_2}$,
  we have that $\varrho(\beta' a_H o)= \varrho(\beta')$.
  It is immediate that $\varrho(\beta' a_H o)$ is not winning for either player, and it
  remains to show that $a_H$ is truthful at $\beta'$.

  As noted above,  the assumption that $\beta'$ is truthful, together with the assumption that $\varrho(\beta')$ is  
  not a winning play for
  either player, implies that $K_H(\beta') = S(c,\nu(\beta'),a)$, by Proposition~\ref{prop:undecided}.
  We will show that the desired conclusion that $a_H$ is truthful at $\beta'$ follows from the
  weaker assumption that $\beta'$ is truthful and $K_H(\beta') = S(c,\nu(\beta'),a)$: this
  helps with the argument for case $c=\bot$, which is handled below.

  Note first that $a_H$ cannot be $\isopen_j$, since the final state of the prefix $r'$ of $r$
  with $\view_H(r') = \beta'$  is in $S(c,\nu(\beta'),a)$, hence in $\C\times\Plates\times\B\times\lastmove$.
  so the action $\isopen_j$ results in a transition to the state $(\nxt(c), \lose, \bot)$ in $r$.
  It follows that the final state of $r_1$ is $(L_2, \lose, \bot)$, and then the subsequent action
  $\cw$ produces a run consistent with $\pi$ with view $\alpha$, contrary to the assumption that $\pi$ excludes $\alpha$.
  Hence $a_H= \isblocked_j$ for some $j$.

  Suppose  that $\Open^2_k(j,\nu(\beta')(j))$. Since $K_H(\beta') = S(c,\nu(\beta'),a)$,
  we have that $(c,j,\nu(\beta')(j),a) \in  K_H(\beta')$. By Proposition~\ref{prop:LH}, there exists a run
  $r'$ ending in state $(c,j,\nu(\beta')(j),a)$ with $\view_H(r') = \beta' $ and $\view_L(r')$ a prefix of $\alpha_1$.
  The transition
  $$(c,j,\nu(\beta')(j),a) \LT{\isblocked_j} (\nxt(c), \lose, \bot)$$
  extends this to a run whose $L$ view remains a prefix of $\alpha_1$, and by following strategy 
  $\pi$ and the remaining $L$ actions in $\alpha_1$ we may continue to  extend to the point where 
  we obtain a run $r_1'$ with $\view_L(r_1')= \alpha_1$ and final state $(L_2, \lose, \bot)$. 
  But then the next $\cw$ step allows a transition to $(\bot, \lose, 2)$, and we obtain a run consistent 
  with $\pi$ with $L$ view $\alpha$, a contradiction. Thus, in fact, we must have not $\Open^2_k(j,\nu(\beta')(j))$, 
  so that $a_H$ is truthful at $\beta'$, as required.

  For the purposes of the next case, we make one further conclusion.
  Note that by definition of the transitions for $\isblocked_j$ at stage $H_k$,
  we get from $K_H(\beta') = S(c,\nu(\beta'),a)$
  that $K_H(\beta' a_H o) = S(\nxt(c),\nu(\beta'),o)= S(\nxt(c),\nu(\beta' a_H o),o)$, so we
  preserve the weakened assumption.
  Thus, since the above argument applies for all $k = 1\ldots h_2$,
  we have that for all such $k$, there exists $j$ such that not $\Open^2_k(j,\nu(\beta')(j))$.
  That is, no hole of player~2 is open in $\nu(\beta')$.
  It follows that $\nu(\beta')$ cannot be a winning position of player 2.

  \item If $c= \bot$, then $a_H$ is trivially truthful, and $\varrho(\beta' a_H o)= \varrho(\beta')\LT{\lambda_m} \nu_m$,
  where $o= \lambda_m$ is a move of player 2. As noted above,
  $K_H(\beta') = S(c,\nu(\beta'),a)$ for some $a$.
  The transitions for case $c= \bot$ then imply that  $K_H(\beta' a_H o)= S(c,\nu(\beta' a_H o),o)$.
  We therefore satisfy the weakened assumption for the stages $H_1\ldots H_{h_2}$ in the previous case.
  It therefore follows  using the argument of the previous case that $\nu(\beta' \, a_H \,o)= \nu_m$ 
  cannot be a winning position of player 2. Thus, from the assumption that $\varrho(\beta')$ is not winning for either player, we
  obtain that $\varrho(\beta' a_H o)$ is not winning for either player.
  \end{enumerate}
This completes the argument that $\beta_1$ is truthful and $\varrho(\beta_1)$ is
not a winning play for either player. By Proposition~\ref{prop:undecided},
we obtain that $K_H(\beta_1) = S(c,\nu(\beta_1),a)$ for some $a$.
In particular, the final state of $r_1$ must be in $\C\times\Plates\times\B\times\lastmove$,
and the next $\cw$ action then results in a run consistent with $\pi$
with $L$ view $\alpha$, a contradiction. \qed
\end{proof}

Again, we point out that the hardness result holds for scheduled
machines already.

\section{Synchronous Bisimulation-based Notions} \label{sec:res}

In this section we establish the result:
\begin{theorem}\label{thm-res} For the class of finite state synchronous machines, $\tRES$ is in PTIME.
\end{theorem}
The following Lemma shows that in searching for an unwinding relation on a machine $M$, it
suffices to consider equivalence relations on the reachable states of $M$.

\begin{lemma}\label{lem:unw} \mbox{}
\begin{enumerate}
\item If there exists a synchronous unwinding relation on $M$, then there
  exists a largest such relation, which is transitive.
\item If all states in $M$ are reachable then the largest synchronous
  unwinding relation  (if one exists) is an equivalence relation.
\item A system satisfies $\tRES$ iff its restriction to its reachable
  states satisfies $\tRES$.
\end{enumerate}
\end{lemma}

\begin{proof} 
\begin{enumerate}
\item First we show that the set of synchronous unwinding relations on
  $M$ is closed under union.  Let $\sim_1$ and $\sim_2$ be two
  unwinding relations on $M$. Clearly Items~1 and~2 of Def.~\ref{def-unwinding}
  hold for $\sim_1 \cup \sim_2$.  Item~3 holds as well as if $s \sim_1
  \cup \sim_2 t$ then either $s \sim_1 t$ or $s \sim_2 t$
  holds. Assume $s \sim_1 t$, then as $\sim_1$ is an unwinding
  relation, by Item~3 of Def.~\ref{def-unwinding}, it follows that for all
  $a_1, a_2 \in A_H$ and $a_3 \in A_L$, if $s \xrightarrow{\
    (a_1,a_3)\ } s'$ there is some $t'$ such that $t \xrightarrow{\
    (a_2,a_3)\ } t'$ and $s' \sim_1 t'$, which implies $s' \sim_1 \cup
  \sim_2 t'$. This implies that there exists a largest unwinding
  relation.

  Second, the composition of two unwinding relations is an unwinding
  relation.  Again Items~1 and~2 of Def.~\ref{def-unwinding} hold for
  $\sim_1 \circ \sim_2$. Assume $s \sim_1 \circ \sim_2 t$. In this
  case there is some $x$ such that $s \sim_1 x$ and $x \sim_2 t$. As
  $\sim_1$ is un unwinding relation, by Item~3 of
  Def.~\ref{def-unwinding}, for any $a_1, a_2 \in A_H$ and $a_3 \in
  A_L$, if $s \xrightarrow{\ (a_1,a_3)\ } s'$, there is some $x'$ such
  that $x \xrightarrow{\ (a_2,a_3)\ } x'$ and $x' \sim_1 t'$.  As $x
  \sim_2 t$, and as $\sim_2$ is an unwinding relation, Item~3
  Def.~\ref{def-unwinding} applied with $a_1=a_2$ implies there exists
  some $t'$ such that $t \xrightarrow{\ (a_2,a_3)\ } t'$ and $x'
  \sim_2 t'$.  Putting it all together, there is some $t'$ such that
  $t \xrightarrow{\ (a_2,a_3)\ } t'$ and some $x'$ such that $s'
  \sim_1 x' \sim_2 t'$ \ie $s' \sim_1 \circ \sim_2 t'$.

  It follows that the transitive closure of any synchronous unwinding relation is
  a synchronous unwinding relation. In particular, the largest such
  relation must be transitive.

\item Let $\sim$ be 
  the largest synchronous unwinding relation. 
  By definition and the Item~1 above, we already have that $\sim$ is
  symmetric and transitive, so it suffices to show reflexivity.  Let
  $s$ be a reachable state. In this case there is a run
  $s_0a_1s_1\dots a_ns_n$ of $M$ such that $s=s_n$.  We need to show
  that $s\sim s$. The proof is by induction on the length of the run.

The base case of $s_0 \sim s_0$ is immediate from Item~1 of
Def.~\ref{def-unwinding}. Suppose $s_i\sim s_i, 0 \leq i \leq n$.
Assume $s_n\xrightarrow{\ a_{n+1}\ }s_{n+1}$. As $s_n \sim s_n$, and
applying Item~3 of Def.~\ref{def-unwinding} to the right handside copy
of $s_n$, there exists a transition $s_n\xrightarrow{\ a_{n+1}\ } s'$
for some $s'$ and $s_{i+1}\sim s'$. Because $\sim$ is symmetric, we
have $s'\sim s_{n+1}$, and by transitivity $s_{n+1}\sim s_{n+1}$.

\item Any synchronous unwinding relation is still a synchronous
  unwinding relation when restricted to the reachable states.
  Conversely, given a synchronous unwinding relation on the reachable
  states, the (identical) relation which extends this to all states by
  union with the empty relation on unreachable states is also a
  synchronous unwinding relation.  \qed
\end{enumerate}
\end{proof}

Using Lemma~\ref{lem:unw}, we can design an algorithm to compute the
largest synchronous unwinding relation, or the empty relation
if none exists. By  part (3) of  Lemma~\ref{lem:unw} we may
assume that all the states of the machine
$M=\word{S,A,s_0,\rightarrow, O, \obs}$ are reachable.

The algorithm is an adaptation of the algorithm for
calculating the relational coarsest partition by Kanellakis and
Smolka~\cite{KS83}. For a partition $P$ or equivalence relation $\approx$,
we write $[s]_{P}$ or $[s]_{\approx}$ for the equivalence class
containing element $s$.
We say that a partition $P$ of the state space $S$ is {\em stable}
if the corresponding equivalence relation $\sim_P$ satisfies condition (3)
of Def.~\ref{def-unwinding}.%
\footnote{Our definition of stability differs from that of
Kanellakis and Smolka: they require that
for each of a set of functions $f_a: S\rightarrow {\cal P}(S)$
(each corresponding to transitions with respect to some label $a$), and partition
$p\in P$, for states $s,t\in S$ we have $s \sim_P t $ implies that that $f_a(s)\cap p = \emptyset$
iff $f_a(t) \cap p\neq \emptyset$. In this condition, we
apply the same function to $s$ and $t$. Our definition amounts to the
application of different functions to $s$ and $t$, corresponding
to transitions with respect to $(a_1,a_3)$ and $(a_2,a_3)$, respectively.}
The idea of the algorithm is to compute the coarsest stable partition satisfying condition (2) of Def.~\ref{def-unwinding},
by iteratively refining an existing partition if the latter is not stable.
 Given the current partition $P$, if there exists a (reachable) state $s$ such that
condition (3) of Def.~\ref{def-unwinding} is not satisfied with $t=s$,
the algorithm terminates and returns the empty relation: this follows
from Lemma~\ref{lem:unw}.(2) because reflexivity is a necessary
condition for the existence of an unwinding relation.
Otherwise, we check whether it is stable for $s\neq t$.
If it is, we have found the largest unwinding relation. If
not, we refine the current partition
based on the counterexample found.

The procedure is given by Algorithm~\ref{alg-unwin}.
In each refinement step, with the current partition equal to $P$, we
first compute the set $$R(s,a_H,a_L) = \{ [t]_P~|~ s\LT{(a_L,a_H)} t\}$$
for each state $s$, and $L$ action $a_L$ and $H$ action $a_H$.
Rule (3) with respect to $\sim_P$ is equivalent to the statement that
if $s \sim_P t$ then $R(s,a_H,a_L) = R(t,a_H',a_L)$ for all
$L$ actions $a_L$ and $H$ actions $a_H,a_H'$.
We first check this when $s=t$. Note that once this condition
has been verified, we have verified that $R(s,a_H,a_L)$ does not
depend on the second argument $a_H$. To check the non-reflexive cases,
it therefore suffices to check the condition with respect to any fixed $a_H$.

\begin{algorithm}[t]
\SetAlgoVlined
\SetNoFillComment
\caption{$\textsc{Compute-Largest-Unwinding}(M)$}
\KwIn{A synchronous machine $M=\word{S,A,s_0,\rightarrow, O, \obs}$. \tcc*[h]{We assume $S$ is reachable }}
\KwOut{The largest unwinding partition $P$ on $M$, or $\emptyset$ if none exists}
\tcc{Initial partition $P$ is given by the set of $L$ observations}
$P \leftarrow \{ \textit{obs}_L^{-1}(o) \ | \ o \in O \}$\;
Loop:
  \tcc{Check stability of each subset $p \in P$}
  \ForEach{$p \in P$}{
      \ForEach{$s\in p$ and $(a_H,a_L)\in A$}{
     Let $R(s,a_H,a_L) = \{ [t]_{P}~|~s\LT{(a_H,a_L)}t\}$\;}
     \tcc{Check reflexive case of condition (3)}
      \ForEach{$s\in p$ and $a_3\in A_L$}{
      \If{there exists $a_1,a_2\in A_H$ with $R(s,a_1,a_3) \neq R(s,a_2,a_3)$}
      {\Return $\emptyset$\;}}
      \tcc{$R(s,a_H,a_L)$ is independent of $a_H$}
      Fix $a_H\in A_H$\;
      For $s,t\in p$ and $a_3 \in A_L$ let  $s\approx_{a_3}  t$ when $R(s,a_H,a_3) = R(t,a_H,a_3)$\;
        \If{there exists $s,t\in p$ and $a_3\in A_L$ with $s\neq t$ and $s\not \approx_{a_3} t$ }
      {
      \tcc{split $p$ according to equivalence classes $[s]_{\approx_{a_3}}$ of $\approx_{a_3}$}
      $P \leftarrow (P\setminus \{p\})\cup \{ [s]_{\approx_{a_3}}~|~s\in p\}$\;
      Go to Loop
      }
      }
       \tcc{$P$ contains the largest stable partition}
  \Return $P$\;
\label{alg-unwin}
\end{algorithm}

We can now prove Theorem~\ref{thm-res}:

\begin{proof}[of Theorem~\ref{thm-res}]
 Algorithm~\ref{alg-unwin} terminates when no split occurs in the main loop.  Since, when a split occurs,
 the new partition is a strict refinement of the previous one, the
 number of iterations of the main loop is at most $|S|$.
 For each $p\in P$, computing the function $R$ can be done in time
 $O((|A_H|\times |A_L| \times |p|) + |p~\cdot\LT{}|)$.  Checking the
 reflexive cases can be done in time $O(|A_H|^2\times |A_L| \times
 |p|)$, and the non-reflexive cases can be done in time $O(|p|^2\times
 |A_L|)$.  Hence the $\mathbf{foreach}$ loop over $p\in P$ can be
 handled in time $O( |A_H|^2\times |A_L|\times|S|^2 + |\LT{}|)$.
 Since there are at most $|S|$ iterations, we have a total time of $O(
 |A_H|^2\times |A_L|\times|S|^3 + |\LT{}|\times |S|)$.  Because
 $|\LT{}|$ may be as large as $|S|^2$, this is $O( |A_H|^2\times
 |A_L|\times|S|^3)$.  (Literature subsequent to Kanellakis and Smolka
 has shown how to optimize their algorithm using careful scheduling,
 union-find data stuctures and amortized analysis, as well as parallel
 implementation.  Similar optimizations may be applicable to our
 algorithm, but we will not pursue this here.)

To argue correctness, we first show that if there exists a synchronous
unwinding $\sim$ on $M$, corresponding to partition $P_\sim$,
the algorithm maintains the invariant  that $P_\sim$ is a refinement of $P$.
That this holds for the initial value of $P$ follows from condition (2) of
Def.~\ref{def-unwinding}. The only case where $P$ changes value is
where we have $p\in P$ and $a_3\in A_L$ with
\begin{enumerate}
\item $R(s,a_1,a_3) = R(s,a_2,a3)$ for all $s\in p$ and $a_1,a_2 \in A_H$,
(there are no reflexivity violations), and
\item $R(s,a_H,a_3) \neq R(t,a_H,a_3)$ for some $s,t\in p$ and $a_H\in A_H$.
\end{enumerate}
In this case, we obtain the new value $P'$ for $P$ by splitting $p$ into
the collection $\{[s]_{\approx_{a_3}}~|~s\in p\}$, where
$\approx_{a_3}$ is defined on $p$ by $s\approx_{a_3} t$ if  $R(s,a_H,a_3) = R(t,a_H,a_3)$.
Suppose that $P_\sim$ is not a refinement of  $P'$. Since the only element of $P$ that changed was $p$,
we must have $s,t\in p$ with $s\sim t$ and $R(s,a_H,a_3) \neq R(t,a_H,a_3)$.
The latter means that there exists $p'\in P$ with (without loss of generality)
$p'\in R(s,a_H,a_3)$ and $p'\not \in  R(t,a_H,a_3)$.
That is, there exists $s'\in p'$ such that $s \LT{(a_H,a_3)}s'$, but
for all $t'$  with $t\LT{(a_H,a_3)} t'$ we have $t'\not \in p'$.
Because $\sim$ is a synchronous unwinding, $s\sim t$, and $s \LT{(a_H,a_3)}s'$,
there exists $t'$ with $t \LT{(a_H,a_3)}t'$ and $s'\sim t'$.
But because $P_\sim$ refines $P$, this implies that $t'\in [s']_P= p'$, a contradiction.
We conclude that in fact $P'$ refines $P$.

The correctness argument now follows straightforwardly.
Suppose that the algorithm outputs $\emptyset$: we show that
there exists no synchronous unwinding on $M$. Suppose to the contrary that $\sim$
is a synchronous unwinding.
At the time the algorithm terminates, we have
$R(s,a_1,a_3) \neq  R(s,a_2,a3)$ for some $p\in P$, some
$s\in p$ and some $a_1,a_2 \in A_H$. Without loss of
generality,
there exists some $p'\in P$ and $t\in S$ such that  $s\LT{(a_1,a_3)} t \in p'$
but there exists no $t'\in p'$ such that $s\LT{(a_2,a_3)} t'\in p'$.
By reflexivity of $\sim$, we have $s\sim s$,
so there exists $t'$ such that $s\LT{(a_2,a_3)} t' $
and $t \sim t'$. Because $P_\sim$ is a refinement of $P$, we obtain
$t' \in [t]_P = p'$, a contradiction.
We conclude that there exists no synchronous unwinding.

Conversely, suppose that the algorithm  outputs a partition $P\neq \emptyset$,
and let $\sim_P$ be the corresponding equivalence relation.
Since $P$ is a refinement of  $ \{ \textit{obs}_L^{-1}(o) \ | \ o \in O \}$,
we have that $s\sim_P t$ implies $\obs_L(s) = \obs_L(t)$,
so condition (2) of Def.~\ref{def-unwinding} is satisfied.
Moreover, we have, for all $p\in P$, that
\begin{enumerate}
\item
$R(s,a_1,a_3) = R(s,a_2,a3)$ for all  $s\in p$ and $a_1,a_2 \in A_H$,
\item
$R(s,a_H,a_3) = R(t,a_H,a_3)$ for all $s\neq t\in p$ and $a_H\in A_H$.
\end{enumerate}
Together, these imply that $\sim_P$ satisfies condition (3) of
Def.~\ref{def-unwinding}. Finally, since $P$ is a partition,
we have $s_0\sim_P s_0$, so condition (1) also holds.
\qed
\end{proof}

\section{Related Work} \label{sec:related}

In asynchronous machines the verification complexities of $\tNDI$ and
$\tNDS$ are both PSPACE-complete, and $\tRES$ (based on asynchronous
unwinding) is in polynomial time~\cite{FG95,FG96,MZ06b}.
Interestingly, PSPACE is also the complexity result for verifying
Mantel's BSP conditions~\cite{Mantel00} on asynchronous finite state
systems. For (asynchronous) push-down systems, the verification
problem is undecidable~\cite{DHKRS08}.

A number of works have defined notions of security for synchronous or
timed systems, but fewer complexity results are known. K{\"o}pf and
Basin \cite{KB06} define a notion similar to $\tRES$ and show it is
PTIME decidable. Similar definitions are also used in the literature
on language-based security \cite{Agat00,VS97}.

Focardi \emph{et al.} \cite{FGM00} define a spectrum of definitions
related to ours in a timed process algebraic setting, and state a
decidability result for one of them, close to our notion
$\tNDS$. However, this result concerns an approximation to the notion
``timed nondeducibility on compositions" (tBNDC)
 that is their real target, and they do not give a complexity
result.  Beauquier and Lanotte defined \emph{covert channels} in timed
systems with $tick$ transitions by using strategies~\cite{BL06}. They
prove that the problem of the existence of a covert channel in such
systems is decidable.  However, their definition of covert channel
requires that $H$ and $L$ have strategies to force a system into
sets of runs with disjoint sets of $L$ views.
The induced definition on \emph{free of covert channels} appears
to be a weaker notion than~$\tNDS$.

\section{Conclusion}
\label{sec:concl}

We remarked above that nondeducibility-based notions of security may
have the disadvantage that they do not readily support a compositional
approach to secure systems development, motivating the introduction of
unwinding-based definitions of security.  The complexity results of
the present paper can be interpreted as lending further support to the
value of unwinding-based definitions.  We have found that the two
nondeducibility notions we have considered, while both decidable, are
intractable.  On the other hand, the unwinding-based notion of
synchronous restrictiveness has tractable complexity. This makes this
definition a more appropriate basis for automated verification of
security. Even if the desired security property is nondeducibility on
inputs or nondeducibility on strategies, it is sufficient to verify
that a system satisfies synchronous restrictiveness, since this is a
stronger notion of security.  It remains to be seen whether there is a
significant number of practical systems that are secure according to
the nondeducibility-based notions, but for which there does not exist
a synchronous unwinding.  If so, then an alternate methodology needs
to be applied for the verification of security for such systems.

\bibliographystyle{spbasic}

\newcommand{\etalchar}[1]{$^{#1}$}

\end{document}